\documentclass{acm_proc_article-sp}
\usepackage{color}
\usepackage{graphicx,caption,subcaption}
\usepackage{amsmath,amssymb,stmaryrd,mathtools}
\usepackage{flushend}
\usepackage{algorithm,algorithmicx}
\usepackage[noend]{algpseudocode}
\usepackage{hyperref}
\usepackage{acronym}

\graphicspath{{./figures/}}

\newtheorem{definition}{Definition}
\newtheorem{problem}{Problem}
\newtheorem{theorem}{Theorem}

\newtheorem{example}{Example}

\newcommand{\ignore}[1]{}

\newcommand{\ie}{i.\,e.,\ }
\newcommand{\etal}{{et~al.}}

\DeclareFontFamily{U}{MnSymbolC}{}
\DeclareSymbolFont{MnSyC}{U}{MnSymbolC}{m}{n}
\DeclareFontShape{U}{MnSymbolC}{m}{n}{
    <-6>  MnSymbolC5
   <6-7>  MnSymbolC6
   <7-8>  MnSymbolC7
   <8-9>  MnSymbolC8
   <9-10> MnSymbolC9
  <10-12> MnSymbolC10
  <12->   MnSymbolC12%
}{}
\DeclareMathSymbol{\powerset}{\mathord}{MnSyC}{180}

\DeclarePairedDelimiter\floor{\lfloor}{\rfloor}

\definecolor{orange}{rgb}{1, .36, .08}

\newcommand{\G}{{\bf G}}
\newcommand{\F}{{\bf F}}
\newcommand{\U}{{\bf U}}

\def\P{{\cal P}} % HA

\newcommand{\s}{{\bf s}}
\newcommand{\xit}{\xi,t}
\newcommand{\uH}{{\bf u}^H}
\newcommand{\wH}{{\bf w}^H}

\newcommand{\x}{{\bf x}}
\newcommand{\uu}{{\bf u}}

\newcommand{\xego}{x^{\mathrm{ego}}}
\newcommand{\yego}{y^{\mathrm{ego}}}
\newcommand{\vego}{v^{\mathrm{ego}}}

\newcommand{\vadv}{v^{\mathrm{adv}}}

\newcommand{\dotxego}{\dot{x}^{\mathrm{ego}}}
\newcommand{\dotyego}{\dot{y}^{\mathrm{ego}}}
\newcommand{\dotvego}{\dot{v}^{\mathrm{ego}}}

\newcommand{\xegot}{x_t^{\mathrm{ego}}}
\newcommand{\yegot}{y_t^{\mathrm{ego}}}
\newcommand{\vegot}{v_t^{\mathrm{ego}}}
\newcommand{\aegot}{a_t^{\mathrm{ego}}}
\newcommand{\xadvt}{x_t^{\mathrm{adv}}}
\newcommand{\yadvt}{y_t^{\mathrm{adv}}}
\newcommand{\vadvt}{v_t^{\mathrm{adv}}}
\newcommand{\aadvt}{a_t^{\mathrm{adv}}}

\newcommand{\fstl}{f^{\mathrm{stl}}}

\newcommand{\fdyn}{f^{\mathrm{dyn}}}

\newcommand{\phie}{\varphi_e}
\newcommand{\phis}{\varphi_s}

\begin{document}

\title{Diagnosis and Repair for Synthesis \\ from Signal Temporal Logic Specifications}

\numberofauthors{1}
\author{
\alignauthor
Shromona Ghosh$^{\S}$\hspace{0.5in} Dorsa Sadigh$^{\S}$\hspace{0.5in} Pierluigi Nuzzo$^{\S}$\\
Vasumathi Raman$^{\dag}$\hspace{0.5in} Alexandre Donz\'{e}$^{\S}$\hspace{0.5in}  Alberto Sangiovanni-Vincentelli$^{\S}$\\
S. Shankar Sastry$^{\S}$\hspace{0.5in} Sanjit A. Seshia$^{\S}$\hspace{0.5in}\\
%\vspace{10pt}
       \affaddr{$^{\S}$Department of Electrical Engineering and Computer Sciences, University of California, Berkeley, CA}\\
       \affaddr{$^{\dag}$United Technologies Research Center, Berkeley, CA}\\
%        \vspace{5pt}
%       \affaddr{Email: \{shromona.ghosh,dsadigh,nuzzo,donze,alberto,sastry,sseshia\}@eecs.berkeley.edu, ramanv@utrc.utc.com}
}
% FIXME: check affiliations and fill in the missing ones

%\CopyrightYear{2016} 
%\setcopyright{acmlicensed}
%\conferenceinfo{HSCC'16,}{April 12 - 14, 2016, Vienna, Austria}
%\isbn{978-1-4503-3955-1/16/04}\acmPrice{\$15.00}
%\doi{http://dx.doi.org/10.1145/2883817.2883847}
\maketitle

\begin{abstract}
%A major challenge for automatic controller synthesis from high-level
%specifications is to come up with good specifications in the first
%place.
We address the problem of diagnosing and repairing specifications
for hybrid systems formalized in signal temporal logic (STL).
Our focus is on the setting of automatic synthesis of controllers in
a model predictive control (MPC) framework.
We build on recent approaches that reduce the controller
synthesis problem to solving one or more mixed integer linear programs
(MILPs), where infeasibility of a MILP usually indicates unrealizability
of the controller synthesis problem.
%Thus, we investigate the infeasibility of the corresponding MILP.
Given an infeasible STL synthesis problem, we present algorithms that provide
feedback on the reasons for unrealizability, and suggestions for making
it realizable. 
%These algorithms leverage the information offered by state-of-the-art solvers about the cause of infeasibility of the MILP associated with the original STL specification, and a corresponding set of feasible relaxations.
Our algorithms are sound and complete, \ie they provide a correct
diagnosis, and always terminate with a non-trivial specification that
is feasible using the chosen synthesis method, when such a solution
exists. We demonstrate the effectiveness of our approach on the
synthesis of controllers for various cyber-physical systems, including an
autonomous driving application and an aircraft
electric power system.
\end{abstract}
%A category including the fourth, optional field follows...
%\category{D.2.8}{Software Engineering}{Metrics}[complexity measures, performance measures]


%\keywords{ACM proceedings, \LaTeX, text tagging}

\section{Introduction}
The automatic synthesis of controllers for hybrid systems from
expressive high-level specification languages allows raising the
level of abstraction for the designer while ensuring correctness of
the resulting controller. In particular, several controller synthesis
methods have been proposed for expressive temporal logics and a
variety of system dynamics. However, a major challenge for the adoption
of these methods in practice is the difficulty of writing 
%complete and consistent 
correctly formal specifications.
Specifications that are poorly
stated, incomplete, or inconsistent can produce
synthesis problems that are unrealizable (no controller exists for the
provided specification), intractable (synthesis is computationally
too hard), or lead to
solutions that fail to capture the designer's intent.
In this paper, we present an algorithmic approach to reduce the
specification burden for controller synthesis from temporal logic
specifications, focusing on the case when the original specification
is unrealizable.

Logical specifications can be provided in multiple ways.
One approach is to provide {\em monolithic} specifications, combining
within a single formula constraints on the environment with desired properties of
the system under control.
%Another approach is to write the specification
%as a {\em contract} between the controller and its environment,
%explicitly separating the assumptions on the environment from the
%guarantees to be provided by the controller~\cite{Nuzzo14,Nuzzo15b}.
In many cases, a system specification can be conveniently provided as
a contract to emphasize what are the responsibilities of the system
under control (guarantees) versus the assumptions on the external, possibly
adversarial, environment~\cite{Nuzzo14,Nuzzo15b}. In such a scenario, besides
\emph{``weakening'' the guarantees},
% as discussed in this section,
realizability of a controller can also be achieved by
\emph{``tightening'' the assumptions}. 
% We introduce a method to diagnose and repair A/G contracts in Section~\ref{sec:contracts}.
Indeed, when the specification is unrealizable, it could be
either because the environment assumptions are too weak, or the
requirements are too strong, or a combination of both.
Finding the ``problem'' with the specification manually can be
a tedious and time-consuming process, nullifying the benefits of
automatic synthesis. Further, in the
{\em reactive} setting, when the environment is adversarial, finding the right assumptions a priori can be difficult.
Thus, given an unrealizable logical specification, there is a need for tools that
localize the cause of unrealizability to (hopefully small) parts of
the formula, and provide suggestions for repairing the formula in an
``optimal'' manner.

The problem of diagnosing and repairing formal requirements
has received its share of attention in the formal methods
community. Ferr{\`{e}}re~\etal~perform diagnosis on faulty executions
of systems with specifications expressed in linear temporal logic
(LTL) and Metric Temporal Logic (MTL)~\cite{traceDiagnostics}. They
identify the cause of unsatisfiability of these properties in the form of
prime implicants, which are conjunctions of literals, and map the failure of a
specification to the failure of these prime implicants. Similar syntax tree based
definitions of unsatisfiable cores for LTL
were presented in~\cite{Schuppan09}.
In the context of synthesis from LTL,
Raman~\etal~\cite{RamanKG13} address the problem of
categorizing the causes of unrealizability, and how to detect them in
high-level robot control specifications.
The use of counter-strategies to derive new environment
assumptions for synthesis has also been much studied over the past few
years~\cite{mining,AMT-fmcad13,li2014synthesis}.
%One of our contributions is to extend the previously studied
%notions of diagnosis and repair to STL
%specifications, and to the setting of controller synthesis in
%an MPC framework.
Our approach, based on exploiting information from optimization
solvers, is similar to that taken by
Nuzzo~\etal~\cite{nuzzo-fmcad10} to extract
unsatisfiable cores for satisfiability
modulo theories (SMT) solving.

In this paper, we address the problem of diagnosing and repairing
specifications formalized in \emph{signal
temporal logic (STL)}~\cite{Maler2004}, a specification language
that is well-suited for hybrid systems.
Our work is conducted in the setting of automated synthesis from STL
using optimization methods in a model predictive control (MPC)
framework~\cite{Raman14,Raman15}.
In this approach to synthesis,
both the system dynamics and the STL requirements 
% (expressed monolithically or as a set of contracts) 
on the system are
encoded as mixed integer linear constraints on variables modeling the
dynamics of the system and its environment.
Controller synthesis is then formulated as an optimization problem
to be solved subject to these constraints~\cite{Raman14}. In the
reactive setting, this approach proceeds by iteratively solving a
combination of optimization problems using a
{\em counterexample-guided inductive synthesis} (CEGIS) scheme~\cite{Raman15}.
In this context, an unrealizable STL specification leads to an infeasible
optimization problem.
We leverage the ability of existing mixed integer linear
programming (MILP) solvers to localize the cause of infeasibility
to so-called {\em irreducibly inconsistent systems} (IIS).
Our algorithms use the IIS to localize the cause of unrealizability to the relevant
parts of the STL specification. Additionally, we give a method for
generating a {\em minimal set of repairs} to the STL specification
such that, after applying those repairs, the resulting specification is
realizable. The set of repairs is drawn from a suitably defined space
that ensures that we rule out vacuous and other unreasonable
adjustments to the specification. Specifically, in this paper, we focus on the numerical parameters in a formula since their specification is often the most tedious and error-prone part. 
% Our completeness results are with respect to such changes. 
Our algorithms are sound and complete, \ie they provide a correct
diagnosis, and always terminate with a reasonable
specification that
is realizable using the chosen synthesis method, when such a repair
exists in the space of possible repairs.

The problem of infeasibility in constrained predictive control schemes has also been widely addressed in the literature, 
e.g., by adopting robust MPC approaches, soft constraints, and penalty functions~\cite{kerrigan2000soft,scokaert1999feasibility,bemporad1999robust}. Rather than tackling general infeasibility issues in MPC, our focus is on providing tools to help debug the controller specification at design time. However, the deployment of robust or soft-constrained MPC approaches can also benefit from our techniques. 
Our use of MILP does not restrict our method to linear dynamical
systems; indeed, we can handle constrained linear and
piecewise affine systems, mixed logical dynamical (MLD)
systems~\cite{BemporadM99}, and certain differentially flat systems.
We demonstrate the effectiveness of our approach on the
synthesis of controllers for a number of cyber-physical systems, including an
autonomous driving application and an aircraft
electric power system.

\ignore{
Our approach generalizes to a variety of high-dimensional systems with timed
temporal logic specifications, including constrained linear and
piecewise affine systems, mixed logical dynamic (MLD)
systems~\cite{BemporadM99}, and certain differentially flat systems
such as quadrotors and car-like vehicles. They also generalize to
other categories of nonlinear systems and specifications, the
expressivity being only limited by the underlying solver.
}
The paper is organized as follows. We begin in Sec.~\ref{sec:Prelim} and~\ref{sec:overview} with preliminaries and a running example.
We formally define the diagnosis and repair problems in Sec.~\ref{sec:probdef}
and describe our algorithms for both monolithic and contract specifications in Sec.~\ref{sec:diagnosis}
and~\ref{sec:contracts}.  In Sec.~\ref{sec:examples} we illustrate our approach on the case studies, and finally conclude in Sec.~\ref{sec:conclusions}.%Finally, we illustrate our approach and then conclude.
%The rest of the paper is organized as follows. After introducing some preliminary notions in Section~\ref{sec:Prelim} and
%% we give a brief overview \emph{hybrid dynamical systems}, \emph{signal temporal logic} and \emph{model predictive control}. We introduce
%our running example in Section~\ref{sec:overview}, we formally define the diagnosis and repair problems in Section~\ref{sec:probdef}. We describe our algorithms for both monolithic specifications and contracts in Sections~\ref{sec:diagnosis} and~\ref{sec:contracts}, respectively. In Section~\ref{sec:examples} we illustrate our approach on the case studies, and finally conclude in Section~\ref{sec:conclusions}.

% \alberto{there is nothing in the introduction that talks about Section 6 on contracts! The reason why we use contracts in this paper has to be given!}
%Diagnosis of temporal properties
%Repair of temporal properties

%%% Local Variables:
%%% mode: latex
%%% TeX-master: "root"
%%% End: __

\section{Preliminaries}
\label{sec:Prelim}
In this section, we introduce preliminaries on hybrid dynamical systems, the specification language \emph{Signal Temporal Logic}, and the \emph{Model Predictive Control} framework.
%%%%%%%%%%%%%%%%%%%%%%%%%%%%%%%%%%%%%%%%%%%%%%%%%%%%%%%%%%%%%%%%%%
\subsection{Hybrid Dynamical Systems}
We consider a continuous-time hybrid dynamical system:
\begin{equation} \label{eq:csys}
\begin{aligned}
	&\dot{x_t} = f(x_t,u_t,w_t)\\
	&y_t = g(x_t,u_t,w_t),
\end{aligned}
\end{equation}
%
% \dorsa{Add caligraphic sets here, and set for Y}
where $x_t \in \mathcal{X} \subseteq (\mathbb{R}^{n_c} \times \{0,1\}^{n_l})$ represent the hybrid (continuous and logical) states at time $t$, $u_t \in \mathcal{U} \subseteq (\mathbb{R}^{m_c} \times \{0,1\}^{m_l})$ are the hybrid control inputs, $y_t \in \mathcal{Y} \subseteq (\mathbb{R}^{p_c} \times \{0,1\}^{p_l})$ are the outputs, and $w_t \in \mathcal{W} \subseteq (\mathbb{R}^{e_c} \times \{0,1\}^{e_l})$ are the hybrid external inputs, including disturbances and other adversarial inputs from the environment.
Using a sampling period $\Delta t > 0$, the continuous-time system~\eqref{eq:csys} lends itself to the following discrete-time approximation:
\begin{equation}
\begin{aligned}
\label{eq:dynamics}
	&x_{k+1} = f_d(x_k,u_k,w_k)\\
	&y_k = g_d(x_k,u_k,w_k),
\end{aligned}
\end{equation}
% The discrete-time system above
where states and outputs evolve according to time steps $k \in \mathbb{N}$, where $x_k = x(\floor{t/\Delta t}) \in \mathcal{X}$.
Given that the system starts at an initial state $x_0 \in \mathcal{X}$, a \emph{run} of the system can be expressed as:
\begin{equation}
\xi = (x_0, y_0, u_0, w_0),(x_1, y_1, u_1, w_1),(x_2, y_2, u_2, w_2), \dots	
\end{equation}
\ie as a sequence of assignments over the system variables $V = (x,y, u, w)$. A run is, therefore, a \emph{discrete-time signal}. We denote $\xi_k= (x_k, y_k, u_k, w_k)$.

Given an initial state $x_0$, a finite horizon input sequence ${\bf u}^H = u_0,u_1,\dotsc,u_{H-1}$, and a finite horizon environment sequence ${\bf w}^H = w_0,w_1,\dotsc, w_{H-1}$, the finite horizon run of the system modeled by the system dynamics in~\eqref{eq:dynamics} is uniquely expressed as:
\begin{equation}
\begin{aligned}
&\xi^H(x_0, {\bf u}^H, {\bf w}^H) =& \\
&(x_0, y_0, u_0, w_0), \dotsc,(x_{H-1},y_{H-1},u_{H-1},w_{H-1}),&
\end{aligned}	
\end{equation}
where $x_1, \ldots, x_{H-1}$, $y_0, \ldots, y_{H-1}$ are computed using~\eqref{eq:dynamics}.
We finally define a finite-horizon cost function $J(\xi^H)$, mapping $H$-horizon trajectories $\xi^H \in \Xi$ to costs in $\mathbb{R}^+$.

%%%%%%%%%%%%%%%%%%%%%%%%%%%%%%%%%%%%%%%%%%%%%%%%%%%%%%%%%%%%%%%%%%
\subsection{Signal Temporal Logic}

\emph{Signal Temporal Logic} (STL) was first introduced as an extension of \emph{Metric Interval Temporal Logic (MITL)} to reason about the behavior of real-valued dense-time signals~\cite{Maler2004}. STL has been largely applied to specify and monitor real-time properties of hybrid systems~\cite{donze2012temporal}.
Moreover, it offers a robust, quantitative interpretation for the satisfaction of a temporal formula~\cite{DonzeM10,donze2013efficient}, as further detailed below.
%\dorsa{Any more citations or is this good?}

An STL formula $\varphi$ is evaluated on a signal $\xi$ at some time $t$. We say $(\xi,t) \models \varphi$ when $\varphi$  evaluates to true for $\xi$ at time $t$. We instead write $\xi \models \varphi$, if $\xi$ satifies $\varphi$ at time $0$. The atomic predicates of STL are defined by inequalities of the form $\mu(\xi (t))>0$, where $\mu$ is some function of signal $\xi$ at time $t$. Specifically, $\mu$ is used to denote both the function of $\xi(t)$ and the predicate. Any STL formula $\varphi$ consists of Boolean and temporal operations on such predicates. The syntax of STL formulae is defined recursively as follows:
\begin{equation}
\varphi ::= \mu \:|\: \neg \mu \:|\: \varphi \wedge \psi \:|\:\G_{[a,b]} \psi \:|\: \F_{[a,b]} \psi \:|\: \varphi \: \U_{[a,b]} \psi,
\end{equation}
where $\psi$ and $\varphi$ are STL formulae, $\G$ is the \emph{globally} operator, $\F$ is the \emph{finally} operator and $\U$ is the \emph{until} operator. Intuitively, $\xi \models \G_{[a,b]} \psi$ specifies that $\psi$ must hold for signal $\xi$ at all times of the given interval, \ie $ t\in [a,b]$. Similarly $\xi \models \F_{[a,b]} \psi$ specifies that $\psi$ must hold at some time $t'$ of the given interval.
Finally, $\xi \models \varphi \:\U_{[a,b]} \psi$  specifies that $\varphi$ must hold starting from the current time until a specific time $t \in [a,b]$ at which $\psi$ becomes true. Formally, the satisfaction of a formula $\varphi$ for a signal $\xi$ at time $t$ is defined as:
%FIXME: citation for LTL^
%
%\small
\begin{equation}
\label{eq:STL}
\begin{array}{lll}
(\xit) \models \mu &\Leftrightarrow & \mu(\xi (t)) > 0\\
(\xit) \models \neg \mu &\Leftrightarrow & \neg((\xit) \models \mu)\\
(\xit) \models \varphi \land \psi &\Leftrightarrow & (\xit) \models \varphi \land (\xit) \models \psi\\
(\xit) \models \F_{[a,b]} \varphi &\Leftrightarrow & \exists t'\in [t+a, t+b], (\xi, t') \models \varphi\\
(\xit) \models \G_{[a,b]} \varphi &\Leftrightarrow & \forall t'\in [t+a, t+b], (\xi,t') \models \varphi\\
(\xit) \models \varphi~\U_{[a,b]}~\psi &\Leftrightarrow & \exists t' \in [t+a,t+b] \mbox{ s.t. } (\xi, t') \models \psi \\
&&\land \forall t'' \in [t,t'], (\xi,t'') \models \varphi.
\end{array}
\end{equation}
%\normalsize
%
The \emph{bound} of an STL formula is defined as the maximum over the sums of all nested upper bounds on the temporal operators of the STL formula. For instance, given $\psi = \G_{[0,20]} \F_{[1,6]} \varphi_1 \wedge \F_{[2,25]} \varphi_2$, the \emph{bound} can be calculated as $N = \max (6+20, 25) = 26$. An STL formula $\varphi$ is \emph{bounded-time} if it contains no unbounded operators.

%%%%%%%%%%%%%%%%%%%%%%%%%%%%%%%%%%%%%%%%%%%%%%%%%%%%%%%%%%%%%%%%%%
\paragraph{Robust Satisfaction}
\label{sec:RobSat}
A \emph{quantitative} or \emph{robust semantics} is defined for an STL formula $\varphi$ by associating it with a real-valued function $\rho^\varphi$ of the signal $\xi$ and time $t$, which
provides a ``measure'' of the margin by which $\varphi$ is satisfied. Specifically, we require
$(\xit) \models \varphi$ if and only if $\rho^\varphi(\xit) >0$.
% and $\neg ((\xit) \models \varphi)$ if
%$\rho^\varphi(\xit) < 0$ (the value 0 being inconclusive).
The magnitude of $\rho^\varphi(\xit)$ can then be interpreted as an estimate of the ``distance'' of
a signal $\xi$ from the set of trajectories satisfying or violating $\varphi$.

Formally, the quantitative semantics is defined as follows:
%
%\small
\begin{equation}
\label{eq:robust_STL}	
\begin{array}{lll}
\rho^\mu(\xit)&=& \mu(\xi(t))\\
\rho^{\neg \mu}(\xit)&=&  -\mu(\xi(t)) \\
\rho^{\varphi \land \psi}(\xit)&=& \min (\rho^{\varphi}(\xit),\rho^{\psi}(\xit)    )\\
\rho^{\G_{[a,b]} \varphi}(\xit) &=& \min_{t'\in [t+a, t+b]}\rho^{\varphi}(\xi,t')\\
\rho^{\F_{[a,b]} \varphi}(\xit) &=& \max_{t'\in [t+a, t+b]}\rho^{\varphi}(\xi,t')\\
\rho^{\varphi{\U}_{[a,b]}\psi}(\xit)&=&  \max_{t'\in [t+a, t+b]} ( \min (\rho^{\psi}(\xi,t'),    \\
&&~~~~~~~~~~~~~~\min_{t'' \in [t,t']} \rho^{\varphi}(\xi,t'')).
\end{array}
\end{equation}
%\normalsize
%
Using the definitions above, the robustness value can then be computed recursively for any STL formula.

%%%%%%%%%%%%%%%%%%%%%%%%%%%%%%%%%%%%%%%%%%%%%%%%%%%%%%%%%%%%%%%%%%
\subsection{Model Predictive Control}

\sloppypar
\emph{Model Predictive Control} (MPC), or \emph{Receding Horizon Control} (RHC), is a well studied hybrid system control method~\cite{morari1993model,garcia1989}. In RHC, at any time step, the state of the system is observed and an optimization is solved over a finite time horizon $H$, given a set of constraints and a cost function $J$. 
When $f$, as defined in~\eqref{eq:dynamics}, is nonlinear, we assume optimization is performed at each MPC step after locally linearizing the system dynamics. 
% however, the prediction of the next step is computed with the original nonlinear dynamics $f$.
For example, at time $t=k$, the linearized dynamics around the current state and time 
% given the sampling time $\Delta t$, and then 
are used to compute an optimal strategy $\uH_*$ over the time interval $[k, k+H-1]$. 
Only the first component of $\uH_*$ is, however, applied to the system, while a similar optimization problem is solved at time $k+1$ to compute a new optimal control sequence along the interval $[k+1, k+H]$  for the model linearized around $t=k+1$. 
% A similar computation occurs at every time step. 
While the global optimality of MPC is not guaranteed, the technique is frequently used and performs well in practice.

%At every run, an MPC scheme optimizes a cost function subject to a set of constraints on the behaviors of the closed-loop system.
In this paper, we use STL to express temporal constraints on the environment and system runs for MPC. We then translate an STL specification into a set of mixed integer linear constraints, as further detailed below~\cite{Raman14,Raman15}. 
% Given a formula of the form $\G_{[0,\infty)}~\varphi$, the optimization problem solved at each iteration over a finite horizon $H$ is:
%
Given a formula $\varphi$ to be satisfied over a finite horizon $H$, the associated optimization problem has the form:
\begin{equation}
\label{eq:opt_simple}
\begin{aligned}
	 &\underset{\uH}{\text{minimize}}
	& \quad J(\xi^H(x_0, \uH))\\
	 &\text{subject to}
	& \quad \xi^H(x_0,\uH) \models \varphi,
\end{aligned}	
\end{equation}
which extracts a control strategy $\uH$ that minimizes the cost function $J(\xi^H)$ over the finite-horizon trajectory $\xi^H$, while satisfying the STL formula $\varphi$ at time step $0$.
In a closed-loop setting, we compute a fresh $\uH$ at every time step $i \in \mathbb{N}$, replacing $x_0$ with $x_i$ in~\eqref{eq:opt_simple}~\cite{Raman14,Raman15}. 
% Details of this MPC scheme can be found in \cite{Raman14,Raman15}
% \pierluigi{Can we probably explicitly emphasize the dynamics in the above optimization problem? They are now implicitly incorporated into the run \ldots}

While~\eqref{eq:opt_simple} applies to systems without uncontrolled inputs, a more general formulation can be provided to account for an uncontrolled disturbance input $\wH$ that can act, in general, adversarially. To provide this formulation, we assume that the specification is given in the form of an STL \emph{assume-guarantee (A/G) contract}~\cite{Nuzzo14,Nuzzo15b} $\mathcal{C} = (V,  \varphi_e, \varphi \equiv \varphi_e \rightarrow \varphi_s)$, where $V$ is the set of variables,
% i.e., $V = X \cup Y \cup U \cup W$,
$\varphi_e$ captures the assumptions (admitted behaviors) over the (uncontrolled) environment inputs $w$, and $\varphi_s$ describes the guarantees (promised behaviors) over all the system variables. A game-theoretic formulation of the controller synthesis problem can then be represented as a \emph{minimax} optimization problem:
\begin{equation}
\label{eq:opt}
\begin{aligned}
	& \underset{\uH}{\text{minimize}}
	& \underset{\wH \in \mathcal{W}^e}{\text{maximize}}
	& & J(\xi^H(x_0, \uH , \wH))\\
	& \text{subject to}
	& \forall\wH \in \mathcal{W}^e
	& & \xi^H(x_0,\uH,\wH) \models \varphi, \\
\end{aligned}	
\end{equation}
where we aim to find a strategy $\uH$ that minimizes the worst case cost $J(\xi^H)$ over the finite horizon trajectory, under the assumption that the disturbance signal $\wH$ acts adversarially. We use $ \mathcal{W}^e$ in~\eqref{eq:opt} to denote the
set of disturbances that satisfy the environment specification $\varphi_e$, i.e., $\mathcal{W}^e = \{{\bf w} \in \mathcal{W}^H | {\bf w} \models \varphi_e\}$.

%%%%%%%%%%%%%%%%%%%%%%%%%%%%%%%%%%%%%%%%%%%%%%%%%%%%%%%%%%%%%%%%%%
\paragraph{Mixed Integer Linear Program Formulation}
\label{sec:milp}
To solve the control problems in~\eqref{eq:opt_simple} and~\eqref{eq:opt} 
% on a discrete-time representation of a dynamical system, 
the STL formula $\varphi$ can be translated into a
set of mixed integer constraints, thus reducing the optimization
problem to a \emph{Mixed Integer Program} (MIP). Specifically, in this paper, we consider control
problems that can be encoded as 
\emph{Mixed Integer Linear Programs} (MILP). 
% These problems encompass, for
% instance, Mixed Logical Dynamic (MLD) systems~\cite{BemporadM99} with
% STL specifications that only include piecewise linear or affine
% predicates~\cite{Raman14,Raman15}. 

The MILP constraints are constructed recursively on the structure of the STL specification,
and express the robust satisfaction value of the formula. We see from
Section~\ref{sec:RobSat} that the robustness value of formulae with
temporal and Boolean operators is expressed as the \emph{min} or
\emph{max} of the robustness values of the operands over time. 
We then demonstrate the encoding of the \emph{min} operator. Given
$min(\rho^{\varphi_1}, \hdots, \rho^{\varphi_n})$, we introduce 
Boolean variables $z^{\varphi_i}$ for $i \in \{1, \hdots, n\}$ and a
continuous variable $p$. The resulting MILP constraints are: 
\begin{equation} \label{eq:bigM}
\begin{aligned}
	p &\leq \rho^{\varphi_i}, \sum_{i=1\hdots n} z^{\varphi_i}\geq 1 \\
	\rho^{\varphi_i} - (1 - z^{\varphi_i}) M &\leq p \leq \rho^{\varphi_i} + (1 -
   z^{\varphi_i})M
\end{aligned}	
\end{equation}
where $M$ is a constant selected to be much larger than $|\rho^{\varphi_i}|$ for all $i$,   
%(known as ``big-M'')
 and $i \in \{1, \hdots, n\}$. The above constraints ensure that $p$ takes the value of the minimum robustness and $z^{\varphi_i} = 1$ if $\rho^{\varphi_i}$ is the minimum. To get the constraints for \emph{max}, we replace $\leq$ by $\geq$ in~\eqref{eq:bigM}.

%Given an STL formula $\varphi$, an integer variable $z^\varphi_t$ can be introduced to encode the validity of $\varphi$ at time $t\:$, \ie $z^\varphi_t = 1$ if and only if $\varphi$ holds true at time $t$. Based on such a set of integer variables, a set of MILP constraints can be recursively generated for the STL specification. At the lowest level of the syntax tree of an STL formula, which is the \emph{predicate} level, there is a binary variable $z^\mu_t$ representing the validity of predicate $\mu(x_t)$.
%To enforce that $z^\mu_t = 1$ if and only if $\mu(x_t) >0$, we adopt the widely used ``big-M'' method~\cite{griva2009linear}, thus obtaining:
%
%\begin{equation}
%\begin{aligned}
%	\mu(x_t) &\leq M_t z^\mu_t - \epsilon_t \\
%	-\mu(x_t) &\leq M_t (1-z_t^\mu) -\epsilon_t,
%\end{aligned}	
%\end{equation}
%
%where $M_t$ is a sufficiently large positive number, and $\epsilon_t$ is a sufficiently small positive number. The rest of the MILP constraints are generated recursively using Boolean and temporal operators on the predicate level constraints~\cite{Raman14}.
%\pierluigi{Consider describing the encoding by using the robustness rather than the integer variables; moreover, it is more instructive to show how we encode the conjunction or disjunction.}

% Thus, in our setting, the controller synthesis problem amounts to solving a mixed integer linear problem, finding a set of control inputs that optimizes a cost function, and satisfies the requirements generated from the STL specifications.
We solve the MILP with an off-the-shelf solver. If the receding horizon scheme is feasible, then %(\emph{realizable})
the controller synthesis problem is {\em realizable}, i.e., the 
algorithm returns a controller that satisfies the specification and optimizes the objective.
However, if the MILP is infeasible,
% (\emph{unrealizable}),
the synthesis problem is {\em unrealizable}. In this case,
the failure to synthesize a controller may well be attributed to just a portion of the STL
specification. 
%Finding this ``cause'' of the failure manually is an \emph{ad hoc} and time-consuming process.
In the rest of the paper we discuss how infeasibility of the MILP constraints can be used to infer the ``cause'' of failure and, consequently,  diagnose and repair the original STL specification.

%%% Local Variables:
%%% mode: latex
%%% TeX-master: "root"
%%% End:

\section{A Running Example}
\label{sec:overview}
\begin{figure}
\centering
\includegraphics[width=0.35\textwidth]{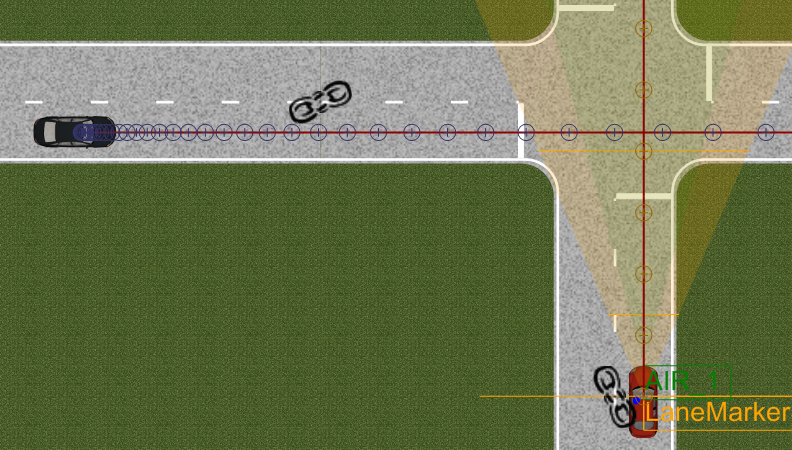}	
\caption{\small{Vehicles crossing an intersection. The red car is the \emph{ego} vehicle, while the black car is part of the environment.}}
\label{fig:intersection}
\end{figure}

To illustrate our approach, we introduce a running example from the autonomous driving domain.
As shown in Fig.~\ref{fig:intersection}, we consider a scenario in which two moving vehicles
approach an intersection.  The red car, labeled the \emph{ego} vehicle, is the vehicle under
control, while the black car is part of the external environment and may behave, in general,
adversarially. The state of the system includes the position and velocity of each vehicle, the
control input is the acceleration of the \emph{ego} vehicle, and the environment input is the acceleration of the other vehicle, i.e.,
\begin{equation}
\begin{aligned}
	 \tilde{x}_t &= (\xegot, \yegot, \vegot, \xadvt, \yadvt, \vadvt)  \\
	 u_t & =  \aegot \quad w_t=\aadvt.
	\end{aligned}	
\end{equation}
We also assume the dynamics of the system is given by a simple double integrator for each vehicle, e.g.,
\begin{equation}
\begin{bmatrix}
	\dotxego \\ \dotyego \\ \dotvego
\end{bmatrix}
 = \begin{bmatrix}
 	0 & 0 & 0\\
 	0 & 0 & 1\\
 	0 & 0 & 0
 \end{bmatrix}
\begin{bmatrix}
	\xego \\ \yego \\ \vego
\end{bmatrix}
+
\begin{bmatrix}
	0 \\ 0 \\ 1
\end{bmatrix}
u.	
\end{equation}
A similar equation holds for the environment vehicle which is, however, constrained to move along
the horizontal axis. We assume the \emph{ego} vehicle is initialized at the coordinates $(0,-1)$ and
the other vehicle is initialized at $(-1,0)$.
% The \emph{ego} vehicle moves along the positive
% y-direction while the other vehicle moves along the positive x-direction.
We further assume all the units in this example
% and all the case studies in this work
follow the metric system. We would like
to design a controller for the \emph{ego} vehicle to satisfy an STL specification under some
assumptions on the external environment, and provide diagnosis and feedback if the specification is
infeasible. We discuss the following three scenarios.
% which were found infeasible for this system and show the
% effectiveness of our approach in fixing them in the following sections.

\begin{example}[Collision Avoidance]
\label{sec:example1}
\begin{sloppypar}
We want to avoid a collision between the \emph{ego} and the adversary vehicle. In this example, we assume the environment vehicle's acceleration is fixed at all times, i.e., $\aadvt = 2$, while the initial velocities are $v_{0}^{\mathrm{adv}} = 0$ and $v_{0}^{\mathrm{ego}} = 0$.
We encode our requirements using the formula
$\varphi := \varphi_1 \wedge \varphi_2,$ where $\varphi_1$ and $\varphi_2$ are defined as follows:
\begin{equation}
\begin{aligned}
% 	&\varphi := \varphi_1 \wedge \varphi_2 ,\\
	&\varphi_1 = \G_{[0,\infty)} \neg\big((-0.5 \leq \yegot \leq 0.5) \wedge (-0.5 \leq \xadvt \leq 0.5) \big) ,\\
	&\varphi_2 = \G_{[0,\infty)} \big(1.5 \leq \aegot \leq 2.5 \big).
\end{aligned}	
\end{equation}
We prescribe bounds on the system acceleration, and state that both cars should never be confined together within a box  of width $1$ around the intersection $(0,0)$ to avoid a collision.
\end{sloppypar}
\end{example}

\begin{example}[Non-adversarial Race]
\label{sec:example2}
% \alex{Double check formulas: is there a missing outside $\G$ on $\psi_1$? If yes we can probably merge the presentation of the two scenarii since only the values allowed by $\aadvt$ is different. }
We discuss a race scenario, in which the \emph{ego} vehicle must increase its velocity to exceed $0.5$ whenever the adversary's initial velocity exceeds $0.5$.  We then formalize our requirement as a contract ($\psi_e$, $\psi_e \rightarrow \psi_s$), where $\psi_e$ are the assumptions made on the environment and $\psi_s$ are the guarantees of the system provided the environment satisfies the assumptions. Specifically:
\begin{equation}
\begin{aligned}
% 	&\psi := \psi_e \rightarrow \psi_s ,\\
	&\psi_e = ( v_0^{\mathrm{adv}} \geq 0.5 ),\\
% 	&\psi_s = \psi_1 \wedge \psi_2 ,\\
	&\psi_s = \G_{[0,\infty)} ( -1 \leq \aegot \leq 1) \land (\vegot \geq 0.5).\\
% 	&\psi_2 = \G_{[0, \infty)} (\vegot \geq 0.5)
\end{aligned}	
\end{equation}
The initial velocities are $v_{0}^{\mathrm{adv}} = 0.55$ and $v_{0}^{\mathrm{ego}} = 0$, while the environment vehicle's acceleration is $\aadvt = 1$ at all times. We also require the acceleration to be bounded by $1$.
% The speeding of the adversary is seen as an environment assumption. The system exhibits two guarantees here, to increase its speed to above $0.5$ ($\psi_2$) and maintain its acceleration in the range $[-1,1]$ ($\psi_2$).
\end{example}

\begin{example}[Adversarial Race]
\label{sec:example3}
\sloppypar
We discuss another race scenario, in which the environment vehicle acceleration $\aadvt$ is no longer fixed, but can vary up to a maximum value of $2$. Initially, $v_{0}^{\mathrm{adv}} = 0$ and $v_{0}^{\mathrm{ego}} = 0$ hold. Under these assumptions, we would like to guarantee that the velocity of the \emph{ego} vehicle exceeds $0.5$ if the speed of the adversary vehicle exceeds $0.5$, while maintaining an acceleration in the $[-1, 1]$ range.
Altogether, we capture the requirements above via a contract ($\phi_w$, $\phi_w \rightarrow \phi_s$), where:
\begin{equation}
\begin{aligned}
% 	&\phi := \phi_e \rightarrow \phi_s ,\\
	&\phi_w = \G_{[0,\infty)} \big( 0 \leq \aadvt \leq 2 \big) ,\\
	&\phi_s =  \G_{[0, \infty)} \big( (\vadvt > 0.5) \rightarrow (\vegot > 0.5) \big) \land \big( |\aegot | \leq 1 \big).
\end{aligned}	
\end{equation}
\end{example}
%\vspace{-20pt}

\section{Problem Statement}
\label{sec:probdef}
In this section, we define the problems of specification diagnosis and repair in the context of controller synthesis from STL. 
We assume the discretized system dynamics $f_d$ and $g_d$, the initial state $x_0$, the STL specification $\varphi$, and a cost function $J$ are given. Then, the \emph{controller synthesis} problem denoted as $\P = (f_d, g_d, x_0,\varphi, J)$ translates into solving~\eqref{eq:opt_simple} (when $\varphi$ is a monolithic specification of the desired system behaviors) or~\eqref{eq:opt} (when $\varphi$ represents a contract between the system and the environment). 

If synthesis fails, the \emph{diagnosis} problem is, intuitively, to return an explanation in the form of a subset of the original problem constraints that are already infeasible when taken alone. 
%Similar concepts have been defined in the context of ``unsatisfiable cores'' and ``unrealizable cores'' for propositional satisfiability and LTL controller synthesis \cite{Schuppan09,RamanKG13,traceDiagnostics}.
The \emph{repair} problem is to return a ``minimal'' set of changes to the specification
% or other parts of the problem description, 
that would render the resulting controller synthesis problem feasible. 
%Similar problems of revision have been posed in, e.g., \cite{Fainekos11,LiDS11}.  Our problem formulation and approach are different from these previous efforts. 
To diagnose and repair an STL formula, we focus on its sets of atomic predicates and time intervals of the temporal operators. We then start by providing a definition of the \emph{support} of its atomic predicates, i.e., the set of times at which the value of a predicate affects satisfiability of the formula, and a notion for the set of repairs that we allow. 
% Here we formally define the diagnosis and repair problems in the context of synthesis from STL specifications.
%
\begin{definition}[Support]
The \emph{support} of a predicate $\mu$ in an STL formula $\varphi$ is the set of times $t$ such that $\mu(\xi (t))$ appears in $\varphi$.
\end{definition}
For example, given $\varphi = \G_{[6,10]}(x_t >0.2)$, the support of predicate $\mu = (x_t >0.2)$ is the time interval $[6,10]$.
%\pierluigi{Probably try a less obvious example?}

\begin{definition}[Allowed Repairs]
Let $\Phi$ denote the set of all possible STL formulae. 
% on the runs of the dynamical system $(f_d, g_d, x_0)$. 
A \emph{repair action} is a relation $\gamma: \Phi \to \Phi$ consisting of the union of the following:
% two relations:  
% \vasu{note that it is a relation, not a function}
\begin{itemize}
% \item A \emph{predicate repair} returns the original formula after replacing one of its atomic predicates, i.e., $\gamma(\varphi) = \varphi[\mu \mapsto \mu + s_\mu]$; 
\item A \emph{predicate repair} returns the original formula after modifying one of its atomic predicates $\mu$ to $\mu^*$. We denote this sort of repair by $\varphi[\mu \mapsto \mu^*] \in \gamma(\varphi)$; 
%\item A \emph{time interval repair} returns the original formula after replacing the interval of a temporal operator: $\gamma(\varphi) = \varphi[\Delta_{[a,b]} \mapsto \Delta_{[a+s_a,b+s_b]}]$;
\item A \emph{time interval repair} returns the original formula after replacing the interval of a temporal operator. This is denoted $\varphi[\Delta_{[a,b]} \mapsto \Delta_{[a^*,b^*]}] \in \gamma(\varphi)$ where $\Delta \in \{\G, \F, \U\}$.
\end{itemize}
% and $s_\mu, s_a, s_b \in \mathbb{R}$.
\end{definition}

%A \emph{repair sequence} $\Gamma$ is a composition of repair actions $\gamma_n(\gamma_{n-1}(\dots(\gamma_1(\varphi))\dots))$. 
\noindent Repair actions can be composed to get a \emph{sequence of repairs} $\Gamma=\gamma_n(\gamma_{n-1}(\dots(\gamma_1(\varphi))\dots))$. 
Given an STL formula $\varphi$, we denote as $\mathtt{REPAIR}(\varphi)$ the set of all possible formulae obtained through compositions of allowed repair actions on $\varphi$. 
%  i.e., $\mathtt{REPAIR}(\varphi) = \{\gamma^*(\varphi)\}$. 
Moreover, given a set of atomic predicates $\cal D$ and a set of time intervals $\cal T$, we use $\mathtt{REPAIR}_{{\cal T,D}}(\varphi) \subseteq \mathtt{REPAIR}(\varphi)$ to denote the set of repair actions that act only on predicates in $\cal D$ or time intervals in $\cal T$.
We are now ready to provide the formulation of the problems addressed in the paper, both in terms of diagnosis and repair of a \emph{monolithic} specification $\varphi$ (\emph{general diagnosis and repair}) and an A/G contract $(\varphi_e, \varphi_e \rightarrow \varphi_s)$ (\emph{contract diagnosis and repair}).
\begin{problem}[General Diagnosis and Repair] \label{prob:diagnosis-simple}
\sloppypar
Given a controller synthesis problem $\P = (f_d, g_d, x_0,\varphi, J)$ such that~\eqref{eq:opt_simple} is infeasible, find: 
\begin{itemize}
\item  A set of atomic predicates $\mathcal{D}=\{\mu_1,\dotsc,\mu_d\}$ or time intervals $\mathcal{T}=\{ \tau_1,\dotsc, \tau_d \}$ of the original formula $\varphi$,
\item $\varphi' \in \mathtt{REPAIR}_{{\cal T,D}}(\varphi)$,
\end{itemize}
such that $\mathcal{P'} = (f_d, g_d, x_0,\varphi', J)$ is feasible, and the following minimality conditions hold:
\begin{itemize}
\item \emph{(predicate minimality)} if $\varphi'$ is obtained by predicate repair\footnote{For technical reasons, our minimality conditions are predicated on a single type of repair being applied to obtain $\varphi'$.}, 
%and $\mathcal{D}=\{\mu_1,\dotsc,\mu_d\}$ is the set of repaired predicates, 
 $s_i = \mu_i^* - \mu_i$ for $i \in \{1,\dotsc,d\}$, $s_{\mathcal{D}} = (s_1,\dotsc,s_d)$, and $|| \cdot ||$ is a norm on $\mathbb{R}^d$, then
\begin{equation}
\begin{aligned}
 &\nexists \: (\mathcal{D'}, s_{\mathcal{D'}}) \quad \emph{s.t.} \quad ||s_{\mathcal{D'}}|| \leq ||s_{\mathcal{D}}||
\end{aligned}
\end{equation}
and $\mathcal{P''} = (f_d, g_d, x_0,\varphi'', J)$ is feasible, with $\varphi'' \in \mathtt{REPAIR}_{{\cal D'}}(\varphi).$
\item \emph{(time interval minimality)} if $\varphi'$ is obtained by time interval repair, 
%and $\mathcal{T} = \{ \tau_1,\dotsc, \tau_d \}$ be the set of time intervals associated to $\mathcal{D}$,
$\mathcal{T}^* = \{\tau_1^*,\dotsc, \tau_l^*\}$ are the non-empty repaired intervals, and $|| \tau ||$ is the length of interval $\tau$:
\begin{equation}
\begin{aligned}
	& \nexists \: \mathcal{T'} = \{\tau_1',\dotsc, \tau_l'\},  \emph{ s.t. } \exists i \in \{1,\dotsc,l \},	||\tau^*_i|| \leq || \tau'_i|| \\
% 	& \wedge \: \exists \: \varphi'' \in \mathtt{REPAIR}_{{\cal T',D}}(\varphi),
\end{aligned}
\end{equation}
and $\mathcal{P''} = (f_d, g_d, x_0,\varphi'', J)$ is feasible, with $\varphi'' \in \mathtt{REPAIR}_{{\cal T'}}(\varphi).$
\end{itemize}
%\item  there do not exist $\mathcal{D}' \subset  \mathcal{D}$ or $\mathcal{T}' \subset  \mathcal{T}$ such that $\P$ becomes feasible by any repair sequence modifying $(\mathcal{D}', \mathcal{T})$ or $(\mathcal{D}, \mathcal{T}')$. 
\end{problem}

\begin{problem}[Contract Diagnosis and Repair]\label{prob:repair-contract}
\sloppypar
Given a controller synthesis problem $\P = (f_d, g_d, x_0,\varphi\equiv\varphi_e\rightarrow\varphi_s, J)$ such that~\eqref{eq:opt} is infeasible, find: 
\begin{itemize}
\item  Sets of atomic predicates $\mathcal{D}_e=\{\mu^e_1,\dotsc,\mu^e_{d}\}$, $\mathcal{D}_s=\{\mu^s_1,\dotsc,\mu^s_{\bar{d}}\}$  or sets of time intervals $\mathcal{T}_e=\{ \tau^e_1,\dotsc, \tau^e_l \}$,$\mathcal{T}_s=\{ \tau^s_1,\dotsc, \tau^s_{\bar{l}} \}$, respectively, of the original formulas $\varphi_e$ and $\varphi_s$,
\item $\varphi_e' \in \mathtt{REPAIR}_{\mathcal{T}_e, \mathcal{D}_e}(\varphi_e)$,  $\varphi_s' \in \mathtt{REPAIR}_{\mathcal{T}_s, \mathcal{D}_s}(\varphi_s)$.
\end{itemize}
such that $\mathcal{P'} = (f_d, g_d, x_0,\varphi', J)$ is feasible, and $\mathcal{D} = \mathcal{D}_e \cup \mathcal{D}_s$, $\mathcal{T} = \mathcal{T}_e \cup \mathcal{T}_s$, and $\varphi'$ satisfy the minimality conditions of Problem~\eqref{prob:diagnosis-simple}.
%Given a controller synthesis problem defined by $\P = (f_d,g_d,x_0,\varphi\equiv\varphi_e\rightarrow\varphi_s,J)$ and such that \eqref{eq:opt} is infeasible, return an STL formula $\varphi' = \varphi'_e\rightarrow\varphi'_s$ with $\varphi_s' \in \mathtt{REPAIR}(\varphi_s)$ and $\varphi_e' \in \mathtt{REPAIR}(\varphi_e)$ such that 
%\begin{itemize}
%\item \eqref{eq:opt} is feasible for the problem $\P' = (f_d,x_0,\varphi',J)$,
%\item
%there do not exist formulas $\varphi_s'' \in \mathtt{REPAIR}(\varphi_s)$, $\varphi_e'' \in \mathtt{REPAIR}(\varphi_e)$  such that $\varphi_s'' \ne \varphi_s'$, $\varphi_e'' \ne \varphi_e'$, $\varphi_s'' \rightarrow \varphi_s'$, $\varphi_e' \rightarrow \varphi_e''$, and \eqref{eq:opt_simple} is feasible for the problem $\P'' = (f_d,g_d, x_0,\varphi_e'' \rightarrow \varphi_s'',J)$.
%\end{itemize}
\end{problem}

% In some situations, it is appropriate to provide weights that dictate which repairs should be preferred. Moreover, we may wish to restrict the set of allowed repairs to a user-defined set of predicates (intervals). 
%We thus also define a weighted repair problem as follows. 
% Moreover, we can restrict the set of allowed repairs to a user-defined set of predicates (intervals) by requiring $s_\mu=0$ ($s_a= s_b = 0$), i.e., no repair action, for all predicates (intervals) not in that set.
%
%\dorsa{Thursday: Go over the weighted stuff}
%\begin{problem}[Weighted Repair]
%Assume given a controller synthesis problem defined by $\P = (f_d, g_d, x_0,\varphi,J)$ such that \eqref{eq:opt_simple} is infeasible. For every predicate $\mu \in \mathcal{D}$ and interval $[a,b] \in \mathcal{T}$, let $w_{\mu}, w_a,w_b$ be a set of weights. For $\varphi' \in \mathtt{REPAIR}(\varphi)$, let $s_\mu = \mu^*-\mu$,  $s_a = a^*-a$, and $s_b = b^*-b$ for each repaired $\mu$ and $[a,b]$. 
%% that appear in $\varphi$, 
%Return $\varphi'$ that solves the relevant \emph{repair} problem (either~\ref{prob:repair-simple} or~\ref{prob:repair-contract}), and additionally minimizes
%$$\sum_{\mu} w_\mu |s_\mu| + \sum_{[a,b]} \left(w_a |s_a| + w_b |s_b| \right).$$
%\end{problem}
%
In the following sections, we discuss our solutions to the above problems. 
%%% Local Variables: 
%%% mode: latex
%%% TeX-master: "root"
%%% End: \input{macros.tex}

\section{Monolithic Specifications}
%\section{General Diagnosis and Repair}
\label{sec:diagnosis}
%%%%%%%%%%%%%%%%%%%%%%%%%%%%%%%%%%%%%%%%%%%%%%%%%%%%%%%%%%%%%

\begin{figure}
\centering
\includegraphics[width=0.55\textwidth]{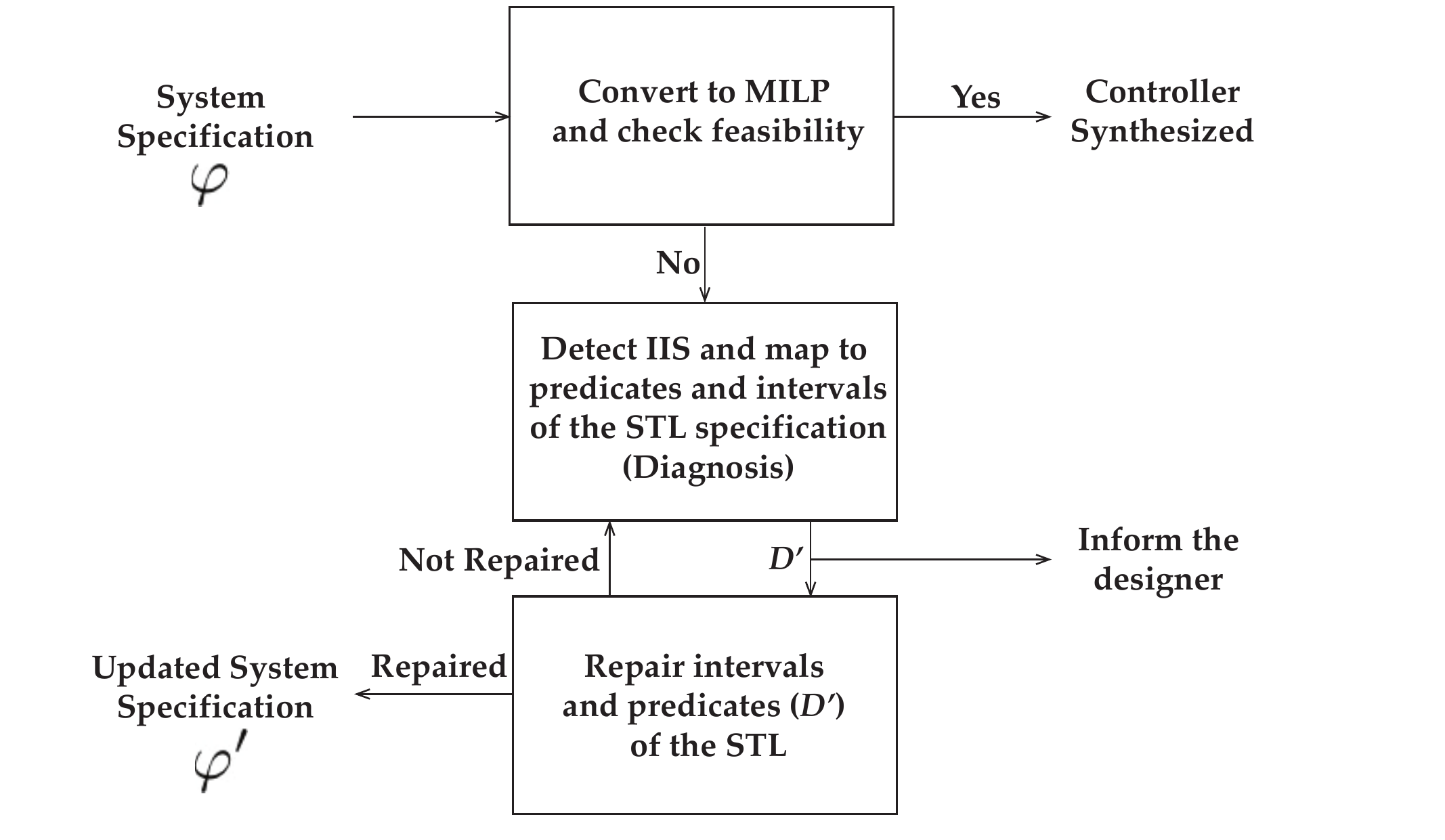}
\caption{\small{Diagnosis and repair flow diagram.}}
%\vspace{-5pt}
\label{fig:diagram}
\end{figure}

\begin{algorithm}[t]
\caption{\small{{\tt DiagnoseRepair}}}\label{algo:DiagnosisRepair}
\scriptsize
\begin{algorithmic}[1]
\medskip
\Procedure{{\tt DiagnoseRepair}}{}
  \vspace{1mm}
 %  \State {$P \gets$ {\tt CreateSynthesisProblem($\psi$, $f_d$, $x_0$)}}
 % \vspace{1mm}
  \State {\textbf{Input:} $\mathcal{P}$}
  \vspace{1mm}
  \State {\textbf{Output: } $\mathbf{u}^H$, $\mathcal{D}$, $repaired$, $\varphi'$}
  \vspace{1mm}
  \State {$ (J, C) \gets$ {\tt GenMILP($\mathcal{P}$)}, $repaired \gets 0$} 
   \vspace{1mm}
  \State {$\mathbf{u}^H \gets$ {\tt Solve($J,C$)}}
  \vspace{1mm}
  \If {$\mathbf{u}^H = \emptyset$}
%  	\vspace{1mm}
%	\State{{\bf return} {\tt SynthesisSuccessful}}
%	\vspace{1mm}
% \Else
  	\vspace{1mm}
	\State {$\mathcal{D} \gets \emptyset$, $\mathcal{S} \gets \emptyset$, $I \gets \emptyset$,  $\mathcal{M} \gets (0, C)$ }
	\vspace{1mm}
	\While {$repaired = 0$}
      %\vspace{1mm}
      %\State{$\mathcal{M}' \gets \mathcal{M}$}
	  \vspace{1mm}
	   \State{$(\mathcal{D}', \mathcal{S}', I') \gets$ {\tt Diagnosis($\mathcal{M}$, $\mathcal{P}$)}}
	  \vspace{1mm}
	  \State{$\mathcal{D} \gets \mathcal{D} \cup \mathcal{D}'$,  $\mathcal{S} \gets \mathcal{S} \cup \mathcal{S}'$, $I \gets I \cup I'$}
	  \vspace{1mm}
	  \State $options \gets $ \tt{UserInput($\mathcal{D}'$)}
	  \vspace{1mm}
	  \State{$\lambda$ $\gets$ {\tt ModifyConstraints($I', options$)}}
	  \vspace{1mm}
	  \State{$(repaired, \mathcal{M},\varphi')\gets{\tt Repair}(\mathcal{M}, I', \lambda, \mathcal{S}, \varphi)$}
	 % \vspace{1mm}
	\EndWhile
	 \vspace{1mm}
	 \State {$\mathbf{u}^H \gets$ {\tt Solve($J$, $\mathcal{M}.C$)}}
  	% \vspace{1mm}
%	\State{{\bf return} $repaired$ }
%	\vspace{1mm}
  \EndIf
 % \State{{\bf return} ($\mathbf{u}^H$, $\mathcal{D}$, $repaired$, $\varphi'$)}
 %  \vspace{1mm}
\EndProcedure 
\medskip
\end{algorithmic}
\end{algorithm}

%\pierluigi{We have just started formalizing the problems, the theorems, and the algorithms.}
The scheme adopted to diagnose inconsistencies in the specification and provide constructive feedback to the designer is pictorially represented in Fig.~\ref{fig:diagram}. In this section we find a solution for Problem~\ref{prob:diagnosis-simple}, as summarized in Algorithm~\ref{algo:DiagnosisRepair}.  Given a problem $\mathcal{P}$, defined as in Section~\ref{sec:probdef}, \texttt{GenMILP} reformulates~\eqref{eq:opt_simple} in terms of the following MILP:
% The STL formula, set of bound constraints on the inputs and the dynamics are encoded as a set of linear equality and inequality constraints.
\begin{equation} \label{opt:plain-diagnosis}
\openup -1\jot
\begin{aligned}
	 \underset{\uH}{\text{minimize}}
	& \quad J(\xi^H) &\\
	 \text{subject to}
	& \quad \fdyn_i \leq 0 & \quad i \in \{ 1,\dotsc, m_d\}\\
	%&\quad \fbound_j \leq 0 & \quad j \in \{ 1,\dotsc, m_i\}\\
	& \quad \fstl_k \leq 0 & \quad k \in \{ 1,\dotsc, m_s\},\\
\end{aligned}	
\end{equation}
where $\fdyn$ and $\fstl$ are mixed integer linear constraint functions over the states, outputs, and inputs of the finite horizon trajectory $\xi^H$ associated, respectively, with the system dynamics and the STL specification $\varphi$. 
We let $(J,C)$ represent this MILP, where $J$ is the objective, and $C$ is the set of constraints.
If problem~\eqref{opt:plain-diagnosis} is infeasible, we iterate between diagnosis and repair phases until the repaired feasible specification $\varphi'$ is obtained. We let $\mathcal{D}$ and $I$ denote, respectively, the set of predicates returned by the diagnosis procedure, and the constraints corresponding to those predicates.

Optionally, we support an interactive repair mechanism, where the designer provides a set of  $options$ that prioritize which predicates to modify ({\tt UserInput} procedure) and get converted into a set of weights $\lambda$  ({\tt ModifyConstraints} routine).
% The overall intent to automatically change a specification, but to pinpoint a set of inconsistencies and repairs to the designer. 
The designer can then leverage this weighted-cost variant of the problem to define ``soft'' and ``hard'' constraints in the controller synthesis problem. In the following, we detail the operation of the \texttt{Diagnosis} and \texttt{Repair} routines.

\subsection{Diagnosis}

\begin{algorithm}[t]
\caption{\small{{\tt Diagnosis}}}\label{algo:Diagnosis}
\scriptsize
\begin{algorithmic}[1]
\medskip
\Procedure{{\tt Diagnosis}}{$\mathcal{M}$, $\mathcal{P}$}
  \vspace{1mm}
  \State {\textbf{Input:} $\mathcal{M}$, $\mathcal{P}$}
  \vspace{1mm}
  \State {\textbf{Output:} $\mathcal{D}$, $\mathcal{S}$, $I'$}
  \vspace{1mm}
  %  \State {$ \mathcal{M} \gets$ {\tt GenMILP($\mathcal{P}$)}} 
  % \vspace{1mm}
  \State {$I_C \gets$ {\tt IIS($\mathcal{M}$)}} 
  \vspace{1mm}
  \State {$(\mathcal{D}, \mathcal{S})  \gets$ {\tt ExtractPredicates($I_C$,$\mathcal{P}$)}}
  \vspace{1mm}
  \State {$I' \gets$ {\tt ExtractConstraints($\mathcal{M}$,$\mathcal{D}$)}}
  \vspace{1mm}
\EndProcedure 
\medskip
\end{algorithmic}
\end{algorithm}
% \vspace{-0.5cm}

%\begin{algorithm}[t]
%\caption{\small{{\tt Diagnosis}}}\label{algo:Diagnosis}
%\scriptsize
%\begin{algorithmic}[1]
%\medskip
%\Procedure{{\tt Diagnosis}}{$M'$, $\mathcal{P}$}
%  \vspace{1mm}
%  \State {\textbf{Input:} $M'$, $\mathcal{P}$}
%  \vspace{1mm}
%  \State {\textbf{Output:} $\psi'$}
%  \vspace{1mm}
%  \State {$I_m \gets$ {\tt IIS($M'$)}} 
%  \vspace{1mm}
%  \State {$\psi_c' \gets$ {\tt ExtractPredicates($I_m$, $\mathcal{P}$)}}
%  \vspace{1mm}
%  \State {$\psi' \gets$ {\tt ConvertToSTL($\psi_c'$)}}
%  \vspace{1mm}
%\EndProcedure 
%\medskip
%\end{algorithmic}
%\end{algorithm}

% \pierluigi{A few changes might be needed in the pseudo-code.} 
%\pierluigi{Probably, we should further comment here on the structure of the constraints and better connect with the prelim section where we introduce the encodings. Will need some connection to the problem definition in Sec. 3; will get back later on.}
%\alberto{Agree}

%We classify the origins of the infeasibilities into three different categories based on the designer's specifications:
%(i) The system dynamics or the initial states have caused the failure.
%(ii) The bounds on the control inputs have caused the infeasibility. 
%(iii) One or some of the STL constraints have failed to be satisfied.
Our diagnosis procedure is summarized in Algorithm~\ref{algo:Diagnosis}.  \texttt{Diagnosis} receives as inputs the controller synthesis problem $\mathcal{P}$ and an associated MILP formulation $\mathcal{M}$. $\mathcal{M}$ can either be the \emph{feasibility problem} associated with the original problem~\eqref{opt:plain-diagnosis}, or a relaxation of it. This feasibility problem has the same, possibly relaxed, constraints as~\eqref{opt:plain-diagnosis} but zero cost. 
Formally, we provide the following definition of relaxed constraint and relaxed optimization problem. 

\begin{definition}[Relaxed Problem]\label{def:relaxed}
We say that a constraint $f' \leq 0$ is a \emph{relaxed version} of $f \leq 0$ if there exists a slack variable $s \in \mathbb{R}^+$ such that $f' = (f - s)$. In this case, we also say that $f \leq 0$ is relaxed into $f' \leq 0$. Then, an optimization problem $\mathcal{O}'$ is a \emph{relaxed version} of another optimization problem $\mathcal{O}$ if it is obtained from $\mathcal{O}$ by relaxing at least one of its constraints. 
\end{definition}

When $\mathcal{M}$ is infeasible, 
% while the system dynamics are assumed as a given, we are still interested in determining whether the reason for infeasibility can be related with the the STL specification. 
we rely on the capability of state-of-the-art MILP solvers to provide an \emph{Irreducibly Inconsistent System} (IIS)~\cite{Gurobi,Chinneck1991} of constraints $I_C$, defined as follows.

\begin{definition}[Irreducibly Inconsistent System]\label{def:iis}
Given a feasibility problem  $\mathcal{M}$ with constraint set $C$, an \emph{Irreducibly Inconsistent System} $I_C$ is a subset of constraints $I_C \subseteq C$ such that: (i) the optimization problem $(0, I_C)$ is infeasible; (ii) $\forall \ c \in I_C$, problem $(0, I_C \setminus \{c\})$ is feasible.   
\end{definition}

In other words, an IIS is an infeasible subset of constraints that becomes feasible if any single constraint is removed. For each constraint in $I_C$, \texttt{ExtractPredicates} traces back the STL predicate(s) originating it, which will be used to construct the set $\mathcal{D}=\{\mu_1, \ldots,\mu_d\}$ of STL atomic predicates in Problem~\ref{prob:diagnosis-simple}, and the corresponding set of support intervals $\mathcal{S}=\{\sigma_1,\ldots, \sigma_d\}$ (adequately truncated to the current horizon $H$), as obtained from the STL syntax tree. $\mathcal{D}$ will be used to produce a relaxed version of $\mathcal{M}$ as further detailed in Section~\ref{sec:repair}. For this purpose, the procedure also returns the subset $I$ of all the constraints in $\mathcal{M}$ that are associated with the predicates in $\mathcal{D}$.

\subsection{Repair}
\label{sec:repair}
%\subsection{Repair}
The diagnosis procedure isolates a set of STL atomic predicates that jointly produce a reason of infeasibility for the synthesis problem. For repair, we are instead interested in how to modify the original formula to make the problem feasible. The repair procedure is summarized in Algorithm~\ref{algo:Repair}. We formulate relaxed versions of the feasibility problem $\mathcal{M}$ associated with problem~\eqref{opt:plain-diagnosis} by using \emph{slack variables}. 
% Slack variables are generally added to the constraints of an optimization problem to change the inequality constraints to equality constraints. For instance, for a linear constraint: $f\leq 0$, the following equivalent constraint is generated: $f +s =0$ , $s\geq 0$.
% We call $s$ a slack variable, and
% if the slack variable $s<0$, the constraint is infeasible. However, if in such cases the constraint $f$ is modified to $f+s$, the infeasibility disappears, and $s$ informs us on how much the constraint needs to be modified.
% So the infeasibilities allow us to diagnose the causes of failure, and the slack variables point to the specific constraint that has created the failure and its value informs the designer about possible fixes of the constraint. 

Let $f_i$, $i \in \{1,\ldots,m\}$ denote both of the categories of 
constraints $\fdyn$ and $\fstl$ in the feasibility problem $\mathcal{M}$.  We reformulate $\mathcal{M}$ into the following \emph{slack feasibility problem}:
% In addition, we split the equality constraints, $f_j$ and inequality constraints, $f_i$. We also assume, $f_i$ and $f_j$ represent all three of the dynamics, $\fdyn$, control bounds, $\fbound$ and the STL specifications, $\fstl$.
%	
\begin{equation}
\label{eq:slack1}
\begin{aligned}
	\underset{\s \in \mathbb{R}^{|I|}}{\text{minimize}}
	& \quad ||\s|| &\\
	 \text{subject to}
	& \quad f_i - s_i \leq 0 & \quad i \in \{ 1,\dotsc, |I|\} \\
	& \quad f_i \leq 0 & \quad i \in \{ |I|+1,\dotsc, m\} \\
	& \quad s_i \geq 0 & \quad i \in \{ 1,\dotsc, |I|\},
\end{aligned}	
\end{equation}
%\begin{equation}
%\label{eq:slack1}
%\begin{aligned}
%	\underset{\s}{\text{minimize}}
%	& \quad \text{norm}(\s) &\\
%	 \text{subject to}
%	& \quad f_i - s_i < 0 & \quad i \in \{ 1,\dotsc, m_{ineq}\}\\
%	& \quad f_j + s_j^p -s_j^n = 0 & \quad j \in \{ 1,\dotsc, m_{eq}\}
%\end{aligned}	
%\end{equation}
where $\s = s_1...s_{|I|}$ is a vector of slack variables added to the subset of optimization constraints $I$, as obtained after the latest call of \texttt{Diagnosis},  to make the problem feasible. Not all the constraints in the original optimization problem~\eqref{opt:plain-diagnosis} can be modified. For instance, the designer will not be able to arbitrarily modify constraints that can directly affect the dynamics of the system, \ie constraints encoded in $\fdyn$. Solving problem~\eqref{eq:slack1} is equivalent to looking for a set of slacks that make the original control problem feasible while minimizing a suitable norm $|| \cdot ||$ of the slack vector. In most of our application examples, we choose the $l_1$-norm, which tends to provide sparser solutions for $\s$, i.e.,~nonzero slacks for a smaller number of constraints. However, other norms can also be used, including weighted norms based on the set of weights $\lambda$. If problem~\eqref{eq:slack1} is feasible, \texttt{ExtractFeedback} uses the solution $\s^*$ to repair the original infeasible specification
$\varphi$. 
% relax Problem~\eqref{opt:plain-diagnosis} and find an optimal control $\uH$. 
Otherwise, the infeasible problem is returned for another round of diagnosis to retrieve further constraints to relax. In what follows, we 
provide details on the implementation of \texttt{ExtractFeedback}.

\begin{algorithm}[t]
\caption{\small{{\tt Repair}}}\label{algo:Repair}
\scriptsize
\begin{algorithmic}[1]
\medskip
\Procedure{{\tt Repair}}{}
  \vspace{1mm}
 \State {\textbf{Input:} $\mathcal{M}$, $I$, $\lambda$, $\mathcal{S}$, $\varphi$}
  \vspace{1mm}
  \State {\textbf{Output:} $repaired$, $\mathcal{M}$, $\varphi$}
  \vspace{1mm}  
  \State {$\mathcal{M}.J \gets \mathcal{M}.J  + \lambda^{\top} s_I$}
  \vspace{1mm}
  \For {$c$ in $I$}
  	\vspace{1mm}
\If {$\lambda(c) > 0$} 
\vspace{1mm}
	\State{$\mathcal{M}.C(c) \gets \mathcal{M}.C(c) + s_c$}
	\vspace{1mm}
\EndIf
%  \vspace{1mm}	
% 	\State {$Obj_{\mathcal{M}} \gets Obj_{\mathcal{M}} + w_{I}[c] * norm(s_c, i)$} \Comment{$i$ decides the type of norm}
%	\vspace{1mm}
  \EndFor
 %  \vspace{1mm}
  \State{($repaired$, $\mathbf{s}^*$) $\gets$ {\tt Solve($\mathcal{M}.J$, $\mathcal{M}.C$)}}
  \vspace{1mm}
  % \State{$changes \gets \emptyset$}
  % \vspace{1mm}
 %  \If {$repaired$ = $1$}
 %  	\vspace{1mm}
\If {$repaired = 1$} 
 \vspace{1mm}
	\State{$\varphi \gets$ {\tt ExtractFeedback}($\mathbf{s}^*$,$\mathcal{S}$,$\varphi$)}
	\vspace{1mm}
\EndIf
%  \vspace{1mm}
%  \State{{\bf return} ($repaired$, $\mathcal{M}$, $\varphi$)}
%  \vspace{1mm}
\EndProcedure 
\medskip
\end{algorithmic}
\end{algorithm}

If a minimum norm solution $\s^*$ can be found, then the slack variables $\s^*$ can be mapped to a set of \emph{predicate repairs} $s_{\mathcal{D}}$, as defined in Problem~\ref{prob:diagnosis-simple}, as follows. The slack vector $\s^*$ in Algorithm~\ref{algo:Repair} includes the set of slack variables $\{s^*_{\mu_i,t}\}$, where $s^*_{\mu_i,t}$ is the variable added to the optimization constraint associated with an atomic predicate $\mu_i \in \mathcal{D}$ at time $t$, $i \in \{1,\ldots,d\}$. We then set 
\begin{equation} \label{eq:predrepsol}
\forall \ i \in \{1,\ldots,d\} \ s_i = \mu_i^* - \mu_i = \max_{t\in \{1,\cdots,H\}} s^*_{\mu_i,t},
\end{equation}
$H$ being the time horizon for~\eqref{opt:plain-diagnosis}, and $s_{\mathcal{D}} = \{s_1, \ldots, s_d \}$. 

To find a set of \emph{time-interval repairs}, we proceed, instead, as follows: 
% Given predicates $\mathcal{D}=\{\mu_1, \ldots,\mu_d\}$ and corresponding support intervals $\mathcal{S}=\{\sigma_1,\ldots, \sigma_d\}$ returned by \texttt{Diagnosis}. 
\begin{enumerate}
	\item The slack vector $\s^*$ in Algorithm~\ref{algo:Repair} includes the set of slack variables $\{s^*_{\mu_i,t}\}$, where $s^*_{\mu_i,t}$ is the variable added to the optimization constraint associated with an atomic predicate $\mu_i \in \mathcal{D}$ at time $t$. For each $\mu_i \in \mathcal{D}$, with support interval $\sigma_i$, we search for the largest time interval $\sigma'_i \subseteq \sigma_i$ such that the slack variables $s^*_{\mu_i,t}$  for $t \in \sigma'_i$  are $0$. If $\mu_i \notin \mathcal{D}$, then we set $\sigma'_i = \sigma_i$.

	\item We convert every temporal operator in $\varphi$ into a combination of $\G$ (timed or untimed) and untimed $\U$ by using the following transformations: 
$$ \F_{[a,b]} \psi = \neg \G_{[a,b]} \neg \psi,$$ 
$$ \psi_1 \U_{[a,b]} \psi_2 = \G_{[0, a]}(\psi_1 \U \  \psi_2) \wedge \F_{[a,b]} \psi_2,$$
where $\U$ is the untimed (unbounded) \emph{until} operator. Let  $\hat{\varphi}$ be the new formula obtained from $\varphi$ after applying these transformations\footnote{While the second transformation introduces a new interval $[0,a]$, its parameters are directly linked to the ones of the original interval $[a,b]$ (now inherited by the $\F$ operator) and will be accordingly processed by the repair routine.}.   
 
	\item The nodes of the parse tree of $\hat{\varphi}$ can then be partitioned into three subsets, $\nu$, $\kappa$, and $\delta$, respectively associated with the atomic \emph{predicates}, \emph{Boolean operators}, and \emph{temporal operators} ($\G, \U$) in $\hat{\varphi}$. We traverse this parse tree from the leaves (atomic predicates) to the root and recursively define for each node $i$ a new support interval $\sigma^*_i$ as follows:
	\begin{equation} \label{eq:tempalgo}
		\sigma^*_i = 
		\begin{cases}
			 \sigma'_i & \mathrm{if} \ i \in \nu \\
			 \underset{j \in C(i)}{\bigcap} \sigma^*_j & \mathrm{if} \ i \in \kappa \cup \delta_{\U} \\
			 \sigma^*_{C(i)} &  \mathrm{if} \ i \in \delta_{\G}  
		\end{cases}
	\end{equation}  
where $C(i)$ denotes the children of node $i$,  while $\delta_{\G}$ and $\delta_{\U}$ are, respectively, the subsets of nodes associated with the $\G$ and $\U$ operators. We observe that the children set of a $\G$ operator node is a singleton. Therefore, with some abuse of notation, we also use $C(i)$ in~\eqref{eq:tempalgo} to denote a single node in the parse tree.   
% $U_u$ and $U_t$ are the untimed and timed until operator.   
	 
	\item We define the interval repair $\hat{\tau}_j$ for each (timed) temporal operator node $j$ in the parse tree of $\hat{\varphi}$ as  $\hat{\tau}_j = \sigma^*_j$.  If $\hat{\tau}_j$ is empty for any $j$,  no time-interval repair is possible. Otherwise, we map back the set of intervals $\{\hat{\tau}_j\}$ into a set of interval repairs $\mathcal{T}^*$ for the original formula $\varphi$ according to the transformations in step 2 and return $\mathcal{T}^*$.  
% \begin{equation}
% 	\tau^*_j = \sigma^*_j 
% \end{equation} 
\end{enumerate}
%
% \pierluigi{Possible comment about the operator nodes that are closer to the predicate nodes in the tree?} 
We provide an example of predicate repair below, while time interval repair is exemplified in Section~\ref{sec:nonadv}.

%%%%%%%%%%%%%%%%%%%%%%%% Running Example 
%As an example, we consider the driving scenario introduced in Section~\ref{sec:example1}. $\fdyn$ encodes the constraints on the dynamics of the two vehicles. Finally, we assume the environment car acceleration has a fixed value $\aadvt = 2$ at all times. Then, we would like to find a controller that satisfies the following STL specification, generating $\fstl$ in the optimization problem:
% and the bounds on the control inputs by $\fbound$. 
% We would thelike to satisfy the following STL constraint, which will be encoded as  MILP constraints:
%\begin{equation}
%\begin{aligned}
%	&\varphi_1  =  \G_{[0,\infty)} \big( (\vadvt < 0.5) \wedge (\vegot < 0.5) \big),\\
%	&\psi_1 = \G_{[0,\infty)} \big ( |u_t| \leq 1 \big )
%\end{aligned}	
%\end{equation}
%
% Here, $\varphi_1$ encodes constraints on the velocity of the two vehicles. This STL formula 
%stating that, at all times, the velocity of the environment vehicle and the \emph{ego} vehicle must be less than $0.5$, and the acceleration of the ego vehicle is bounded by $1$. 
% which expresses a safety bound on the velocity of the two vehicles at all times. 
\begin{example}[Collision Avoidance]
\sloppypar
We diagnose the specifications introduced in Example~\ref{sec:example1}. 
% $\fdyn$ encodes the constraints on the dynamics of the two vehicles. We would like to find a controller that satisfies $\varphi$, generating $\fstl$ in the optimization problem.
To formulate the synthesis problem, we assume a horizon $H = 10$ and a discretization step $\Delta t = 0.2$.  
% there are $6$ time steps that we would need to consider at every step of the optimization. 
The system is found infeasible at the first MPC run, and \texttt{Diagnosis} detects the infeasibility of $\varphi_1 \land \varphi_2$ at time $t = 6$. Intuitively, given the allowed range of accelerations for \emph{ego}, both the cars end up with entering the forbidden box at some time. 
% that there is always an optimal $\uH$ for satisfaction of the second predicate $\vegot < 0.5$. Here, $\varphi_1$ is infeasible, and the IIS only gives information about changing the first predicate. 
%In fact, while $\vegot$ is under control, given the value assumed for $\aadvt$, there is no way for a controller to satisfy the specification. 
% directly affected by our choice of $\uH$ and the dynamics, and it can always be kept under some bound with the optimal control; however, the value of $\aadvt$ is part of the design process.
Algorithm~\ref{algo:DiagnosisRepair} chooses to repair $\varphi_1$ by adding slacks to all of its predicates, such that $\varphi'_1 = (-0.5 - s_{l1} \leq \yegot \leq 0.5 + s_{u1}) \land (-0.5 - s_{l2} \leq \xadvt \leq 0.5 + s_{u2})$. Table~\ref{table:slack1} shows the optimal slack values at each $t$, while  $s_{u1}$ and $s_{l2}$ are set to zero at all $t$. 
% for $2H = 2 \times 3$ time steps. These slack variables are calculated for a single open-loop horizon.
%
%
% From the values of IIS and the slack variables, 
We can then conclude that the specification replacing $\varphi_1$ with $\varphi'_1$
\begin{equation}
	\varphi'_1 = \G_{[0,\infty)} \neg \big((-0.24 \leq \yegot \leq 0.5) \wedge (-0.5 \leq \xadvt \leq 0.43) \big) 
\end{equation}  
is feasible, i.e., the cars will not collide, but the original requirement was overly demanding. 
 
Alternatively, the user can choose to run the repair procedure on $\varphi_2$ and change its predicate as $(1.5 - s_{l} \leq \aegot \leq 2.5 + s_{u})$. In this case, we decide to stick with the original requirement on collision avoidance, and tune, instead, the control ``effort'' to satisfy it. Under the assumption of constant acceleration (and bounds), the slacks will be the same at all $t$. We then obtain $[s_{l}, s_{u}]$ = $[0.82, 0]$, which ultimately turns into
% \begin{equation}
$\varphi'_2 = \G_{[0,\infty)} \big(0.68 \leq \aegot \leq 2.5 \big).$ The \emph{ego} vehicle should then slow down to prevent entering the forbidden box at the same time as the other car. This latter solution is, however, suboptimal with respect to the $l_1$-norm selected in this example. 
% \end{equation}
\end{example}

\begin{table}
\begin{center}
\scalebox{0.9}{
\begin{tabular}{ c|c|c|c|c|c|c|c|c|c|c}
% \hline
 time &  0 & 0.2 & 0.4 & 0.6 & 0.8 & 1 & 1.2 & 1.4 & 1.6 & 1.8\\
  \hline
  $s_{l1}$ & 0 & 0 & 0 & 0 & 0 & -0.26 & 0 & 0 & 0 & 0 \\
  \hline
    $s_{u2}$ & 0 & 0 & 0 & 0 & 0 & 0 & -0.07 & 0 & 0 & 0 \\
%   \hline
\end{tabular}}
\end{center}
%\vspace{-10pt}
\caption{\small{Slack variables for horizon, with $\Delta t = 0.2$, and $H=10$.}}
\label{table:slack1}
\end{table}

%Alternatively, as we note that the slack variables are nonzero starting from the third time step, another possible modification of the original formula would be to only consider $t \in [0, 0.2 - \epsilon]$, $\epsilon \in \mathbb{R}_+$ being a conservative margin to accommodate errors due to time discretization. Then, the newly modified STL specification can assume one of the following forms:
%\dorsa{add parse tree, and consider 4 possible changes, G->F , and -> or}
%\begin{equation}
%\begin{aligned}
%	\varphi^{\text{new}_t}_1  &=&  \G_{[0,0.18]} \big( (\vadvt < 0.5) \wedge (\vegot < 0.5) \big )\\	
%	\varphi^{\text{new}_p}_1  &=&  \G_{[0,\infty)} \big( (\vadvt < 2.01) \wedge (\vegot < 0.5) \big ),
%\end{aligned}	
%\end{equation}
%where we either modify a predicate or the interval of a temporal operator. We observe that feasibility is only guaranteed for the first optimization run. 
%\pierluigi{Would like to discuss the running example and the guarantees we offer on open loop with respect to closed loop.} 
% For the choice of $\epsilon = 0.02$, $\varphi_1^{\text{new}_t}$ is a time-modified STL specification, where the time bounds have changed, and $\varphi_1^{\text{new}_p}$ is a scale-modified STL specification, where the bounds of the predicate have changed so the specifications hold true.

Our algorithm offers the following guarantees, for which a proof is reported below. 
\ignore{The complete proofs can be found in an extended version of our paper~\cite{}.
}
%
% an extended version of this paper with the complete proofs is available online\footnote{\small{https://www.eecs.berkeley.edu/$\sim$dsadigh/DiagnosisAndRepair}}.
%
\begin{theorem}[Soundness] \label{thm:SoundnessGeneral}
Given a controller synthesis problem $\P = (f_d, g_d, x_0, \varphi, J)$, such that~\eqref{eq:opt_simple} is infeasible at time $t$, let $\varphi' \in \texttt{REPAIR}_{\mathcal{D},\mathcal{T}}(\varphi)$ be the repaired formula returned from Algorithm~\ref{algo:DiagnosisRepair} without human intervention,
% ~\ie $w={\bf 1}$, 
for a given set of predicates $\mathcal{D}$ or time interval $\mathcal{T}$. Then, $\P' = (f_d, g_d, x_0, \varphi', J)$ is feasible at time $t$ and ($\varphi'$, $\mathcal{D}$, $\mathcal{T}$) satisfy the minimality conditions in Problem~\ref{prob:diagnosis-simple}.
%Under the assumption that all $w=1$:
%\begin{itemize}
%\item  $\varphi'$ is feasible.	
%\item  Under the assumption that $\varphi'$ is obtained only by predicate repair, and $\mathcal{D}=\{\mu_1,\dotsc,\mu_d\}$ is the set of repaired predicates, we let $s_i = \mu_i^* - \mu_i$ for $i \in \{1,\dotsc,d\}$ and $s_{\mathcal{D}} = (s_1,\dotsc,s_d)$. Then, 
%\begin{equation}
% \nexists \: (\mathcal{D'}, s_{\mathcal{D'}}) \quad \emph{s.t.} \quad ||s_{\mathcal{D'}}|| \leq ||s_{\mathcal{D}}||.
%\end{equation}
%\item Under the assumption that $\varphi'$ is obtained only by time interval repair, and $\mathcal{T} = \{ \tau_1,\dotsc, \tau_d \}$ be the set of time intervals associated to $\mathcal{D}$, we let $\mathcal{T}^* = \{\tau_1^*,\dotsc, \tau_d^*\}$ be the repaired intervals. Then,
%\begin{equation}
%\begin{aligned}
%	\nexists \: \mathcal{T'} = \{\tau_1',\dotsc, \tau_d'\},  \quad \emph{s.t.} \quad ||\tau^*_i|| \leq || \tau'_i||, \\
%	\quad \emph{for any} \quad i \in \{1,\dotsc,d \}.
%\end{aligned}
%\end{equation}
%\end{itemize}
\end{theorem}

\begin{proof}[(Theorem~\ref{thm:SoundnessGeneral})]
	
Suppose $\mathcal{M}$ is the MILP encoding of $\mathcal{P}$ as defined in~\eqref{opt:plain-diagnosis},  $\varphi'$ is the repaired formula, and $\mathcal{D}$ the set of diagnosed predicates, as returned by Algorithm~\ref{algo:DiagnosisRepair}. We start by discussing the case of predicate repair. We let $\mathcal{M}'$ be the MILP encoding of $\mathcal{P}'$ and $\mathcal{D}^* \subseteq \mathcal{D}$ be the set of predicates that are fixed to provide $\varphi'$, i.e., such that $s = (\mu^* - \mu) \neq 0$, with $\mu \in  \mathcal{D}$.
% $\P' = (f_d, g_d, x_0, \varphi', J)$. 
Algorithm~\ref{algo:DiagnosisRepair} modifies $\mathcal{M}$ by introducing a slack variable $s_{\mu,t}$ into each constraint associated with an atomic predicate $\mu$ in $\mathcal{D}$ at time $t$. Such a transformation leads to a feasible MILP $\mathcal{M}''$ and an optimal slack set $\{s^*_{\mu,t} | \mu \in \mathcal{D}, t \in \{1,\ldots,H\}\}$. 
% 
% We define $s_i = \max_{t\in \{1,\cdots,H\}} s_{it}$, and $\mathcal{S}_{\mathcal{D}} = \{ s_i | s_i \neq 0\}$. Then, $\mathcal{D}^*$ is the set of $\mu_i$ such that its corresponding slack variables $s_i \in \mathcal{S}_{\mathcal{D}}$. 
% 
We now observe that $\mathcal{M}'$ and $\mathcal{M}''$ are both a relaxed version of $\mathcal{M}$. In fact, we can view $\mathcal{M}'$ as a version of  $\mathcal{M}$ in which only the constraints associated with the atomic predicates in $\mathcal{D}^*$ are relaxed. Therefore, each constraint having a nonzero slack variable in $\mathcal{M}''$ is also relaxed in $\mathcal{M}'$. Moreover, by~\eqref{eq:predrepsol}, the relaxed constraints in $\mathcal{M}'$ are offset by the largest slack value over the horizon $H$. Then, because $\mathcal{M}''$ is feasible, $\mathcal{M}'$, and subsequently $\mathcal{P}'$, are feasible. 

We now prove that ($\varphi'$,$\mathcal{D}$) satisfy the predicate minimality condition of Problem~\ref{prob:diagnosis-simple}.
% We recall that $\varphi'$ is obtained by modifying predicates $\mu_i \in \mathcal{D}$ to a new predicate $\mu_i' = \mu_i - s_i$, where $s_i$ is defined above. 
Let $\tilde{\varphi}$ be any formula obtained from $\varphi$ after repairing a set of predicates $\tilde{\mathcal{D}}$ such that the resulting problem $\tilde{\mathcal{P}}$ is feasible. We recall that, by Definition~\ref{def:iis}, at least one predicate in $\mathcal{D}$ generates a conflicting constraint and must be repaired for $\mathcal{M}$ to become feasible. Then, $\tilde{\mathcal{D}} \cap \mathcal{D} \neq \emptyset$ holds. Furthermore, since Algorithm~\ref{algo:DiagnosisRepair} iterates by diagnosing and relaxing constraints until feasibility is achieved, $\mathcal{D}$ contains all the predicates that can be responsible for the infeasibility of $\varphi$. In other words, Algorithm~\ref{algo:DiagnosisRepair} finds all the IISs in the original optimization problem and allows relaxing any constraint in the union of the IISs. Therefore, repairing any predicate outside of $\mathcal{D}$ is redundant: 
% If $\mathcal{D}''\setminus \mathcal{D} \neq \emptyset$, then we only need to consider the reduced set $\bar{\mathcal{D}} = \mathcal{D}'' \cap \mathcal{D}$ for repair. Then, 
a predicate repair set that only relaxes the constraints associated with predicates in $\bar{\mathcal{D}}= \tilde{\mathcal{D}} \cap \mathcal{D}$, by the same amount as in $\tilde{\varphi}$, and sets to zero the slack variables associated with predicates in $\mathcal{D}\setminus \bar{\mathcal{D}}$ is also effective and exhibits a smaller slack norm. Let $s_{\bar{\mathcal{D}}}$ be such a repair set and $\bar{\varphi}$ the corresponding repaired formula. $s_{\bar{\mathcal{D}}}$ and 
$s_{\mathcal{D}}$ can then be seen as two repair sets on the same predicate set. However, by the solution of Problem~\eqref{eq:slack1}, we are guaranteed that $s_{\mathcal{D}}$ has minimum norm; then,  $ || s_{\mathcal{D}} || \leq  || s_{\bar{\mathcal{D}}} ||$ will hold for any such formulas $\bar{\varphi}$, and hence $\tilde{\varphi}$. 

We now consider the MILP formulation $\mathcal{M}'$ associated with $\mathcal{P}'$ and $\varphi'$ in the case of time-interval repairs. For each atomic predicate $\mu_i \in \mathcal{D}$, for $i \in \{1, \ldots, |\mathcal{D}|\}$, $\mathcal{M}'$ includes only the associated constraints evaluated over time intervals $\sigma'_i$ for which the slack variables $\{s_{\mu_i,t}\}$ are zero. Such a subset of constraints is trivially feasible. 
%  in both $\mathcal{M}''$, the relaxed version obtained from Algorithm~\ref{algo:DiagnosisRepair} and $\mathcal{M}'$. 
All the other constraints, 
% both in $\mathcal{M}'$ and $\mathcal{M}''$, 
enforcing the satisfaction of Boolean and temporal combination of the atomic predicates in $\varphi'$ are also feasible if the atomic predicate constraints are feasible. Then, $\mathcal{M}'$ is feasible.
%\begin{enumerate}
%\item Consider $\mathcal{M}''$ and $\mathcal{M}'$
%\item For each set of constraints associated to an atomic predicate $k$	in $\mathcal{D}$ (in $\mathcal{M}''$), $\mathcal{M}'$ would have a subset of them corresponding to intervals $[i,j]$ such that the slack of the predicate $s_{k,i} = 0,\dotsc, s_{k,j} = 0$
%\item Because such a subset is trivially feasible in $\mathcal{M}''$ (even in $\mathcal{M}$).
%\item For all constraints associated with Boolean operators or temporal operators in $\mathcal{M}'$ will be function of item 2. Thus they are trivially feasible.
%\end{enumerate}

To show that ($\varphi'$,$\mathcal{T}$) satisfy the minimality condition in Problem~\ref{prob:diagnosis-simple}, we observe that, by the transformations in step 2 of the time-interval repair procedure, $\varphi$ is logically equivalent to a formula $\hat{\varphi}$ which only contains \emph{untimed} $\U$ and \emph{timed} $\G$ operators. Moreover, 
% Note that any formula with \emph{eventually} and \emph{bounded until} operators can be rewritten without any change to its time intervals by only using Boolean operators such as negation, \emph{globally}, and \emph{unbounded until} operators.
$\hat{\varphi}$ and $\varphi$ have the same interval parameters. Therefore, if the proposed repair set is minimal for 
$\hat{\varphi}$, this will also be the case for $\varphi$.  
We now observe that Algorithm~\ref{algo:DiagnosisRepair} selects, for each atomic predicate $\mu_i \in \mathcal{D}$ the largest interval $\sigma'_i$ such that the associated constraints are feasible, i.e., their slack variables are zero after norm minimization\footnote{Because we are not directly maximizing the sparsity of the slack vector, time-interval minimality is to be interpreted with respect to slack norm minimization. Directly maximizing the number of zero slacks is also possible but computationally more intensive.}. Because feasible intervals for Boolean combinations of atomic predicates are obtained by intersecting these maximal intervals, and then propagated to the temporal operators, the length of the intervals of each $\G$ operator in $\hat{\varphi}$, hence of the temporal operators in $\varphi$, will also be maximal, which is what we wanted to prove.    
% Let's convert the formula to $\bar{\varphi}$ a formula with only globally and unbounded until operator on it. The eventually operators are converted to globally with extra negations.
% Both $\bar{\varphi}$ and $\varphi$ have the same intervals.
% We only reason about intervals related to the globally operator.
% We have a set of predicates and constraints related to $\varphi$.
% Given a norm, for some of the constraints in the interval we have a zero slack variable. Given norm 1, you choose the largest zero slack interval, and since this is norm 1, and because of its sparsity, we will have the largest time interval. So We are minimal!
\end{proof}
%}

\begin{theorem}[Completeness] \label{thm:CompletenessGeneral}
\sloppypar
Assume the controller synthesis problem $\P = (f_d, g_d, x_0, \varphi, J)$ results in~\eqref{eq:opt_simple} being infeasible at time $t$.
If there exist a set of predicates $\mathcal{D}$ or time-intervals $\mathcal{T}$ such that there exists $\Phi \subseteq \texttt{REPAIR}_{\mathcal{D},\mathcal{T}}(\varphi)$ for which $\forall \ \phi \in \Phi$, $\P' = (f_d, g_d, x_0, \phi, J)$  is feasible at time $t$ and ($\phi$, $\mathcal{D}$, $\mathcal{T}$) are minimal in the sense of Problem~\ref{prob:diagnosis-simple}, then Algorithm~\ref{algo:DiagnosisRepair} returns a repaired formula $\varphi'$ in $\Phi$.  
\end{theorem}

%\ignore{
\begin{proof}[(Theorem~\ref{thm:CompletenessGeneral})]
We first observe that Algorithm~\ref{algo:DiagnosisRepair} always terminates with a feasible solution $\varphi'$ since the set of MILP
constraints to diagnose and repair is finite. We first consider the case of predicate repairs.  Let $\mathcal{D}$ be the set of predicates modified to obtain  $\phi \in \Phi$ and $\mathcal{D}'$ the set of diagnosed predicates returned by Algorithm~\ref{algo:DiagnosisRepair}. Then, by Definition~\ref{def:iis} and the iterative approach of Algorithm~\ref{algo:DiagnosisRepair}, we are guaranteed that $\mathcal{D}'$ includes all the predicates responsible for inconsistencies, as also argued in the proof of Theorem~\ref{thm:SoundnessGeneral}. Therefore, we conclude $\mathcal{D} \subseteq \mathcal{D}'$. $s_{\mathcal{D}}$ and $s_{\mathcal{D}'}$ can then be seen as two repair sets on the same predicate set. However, by the solution of Problem~\eqref{eq:slack1}, we are guaranteed that $s_{\mathcal{D}'}$ has minimum norm; then,  $ || s_{\mathcal{D}'} || \leq  || s_{\bar{\mathcal{D}}} ||$ will hold, hence $\varphi' \in \Phi$. 

We now consider the case of time-interval repair. 
% Based on time interval minimality defined in Problem~\ref{prob:diagnosis-simple}, 
If a formula $\phi \in \Phi$ repairs a set of intervals $\mathcal{T} = \{\tau_1,\dotsc,\tau_l \}$, then there exists a set of constraints associated with atomic predicates in $\varphi$ which are consistent in $\mathcal{M}$, the MILP encoding associated with $\phi$, and make the overall problem feasible.  Then, the relaxed MILP encoding $\mathcal{M}'$ associated with $\varphi$ after slack norm minimization will also include a set of predicate constraints admitting zero slacks over the same set of time intervals as in $\mathcal{M}$, as determined by $\mathcal{T}$. Since these constraints are enough to make the entire problem $\mathcal{M}$ feasible, this will also be the case for $\mathcal{M}'$. Therefore, our procedure for time-interval repair terminates and produces a set of non-empty intervals $\mathcal{T}' = \{\tau'_1,\dotsc,\tau_l'\}$. Finally, because Algorithm~\ref{algo:DiagnosisRepair} finds the longest intervals for which the slack variables associated with each atomic predicate are zero, we are also guaranteed that $||\tau'_i || \geq  ||\tau_i ||$ for all $i \in \{1,\ldots, l\}$, as also argued in the proof of Theorem~\ref{thm:SoundnessGeneral}. We can then conclude that $\varphi' \in \Phi$ holds.  
% affects the same time interval as the ones affected by $\varphi'$, i.e., $\mathcal{T'} = \{\tau'_1,\dotsc,\tau_l'\}$. 
% Then, let $\tau_i \in \mathcal{T}$, and $\tau_i' \in \mathcal{T'}$ such that . Since $\tau'_i$ is determined by , and as part of the algorithm, it takes the maximum over the set of all zero slack variables, this inequality is not possible, and by proof by contradiction, we conclude that repairs using timing repairs is complete.
\end{proof}
%}

\ignore{
\begin{proof}[Sketch] 
% The inner while loop$(lines 8-14)$ in Algorithm~\eqref{algo:DiagnosisRepair} is the core of our diagnosis and repair procedure. 
We first observe that Algorithm~\ref{algo:DiagnosisRepair} always
terminates with a feasible solution $\varphi'$ since the set of  MILP
constraints is finite. Both the results above will then follow from
the minimality of the IIS and the minimization of the slack vector
norm. In fact, at every iteration, the set of predicates
$\mathcal{D}'$ returned by Algorithm~\ref{algo:Diagnosis} is minimal
since it comes from a minimal irreducible inconsistent system of
constraints. Because Algorithm~\ref{algo:Repair} solves
Problem~\eqref{eq:slack1}, minimality of predicate repairs, i.e.,
predicate slack norm, is guaranteed. Similarly, because of norm
minimization, which also encourages sparsity in the introduced slacks,
Algorithm~\ref{algo:Repair} is able to find the maximum-size
(non-empty) interval $\tau^{*}$ for each temporal operator in
$\varphi$, if this exists, thus guaranteeing minimality of time
interval repairs.  
\end{proof}
}

In the worst case, Algorithm~\ref{algo:DiagnosisRepair} solves a number of MILP problem instances equal to the number of atomic predicates in the STL formula. While the complexity of solving a MILP is NP-hard, the actual runtime depends on the size of the MILP, which is linear in the number of predicates and operators in the STL specification.
%\pierluigi{We can add here or in the case study sections insight on the actual real time in our experiments. The same for Algorithm 4.}

% 
%%% Local Variables: 
%%% mode: latex
%%% TeX-master: "root"
%%% End:  

\section{Contracts}
%\section{Contract Diagnosis and Repair}
\label{sec:contracts}
%%%%%%%%%%%%%%%%%%%%%%%%

In this section, we consider specifications provided in the form of a
contract $(\phie, \phie \rightarrow \phis)$, where $\phie$ is an STL
formula expressing the assumptions, i.e., the set of behaviors assumed
from the environment, while $\phis$ captures the guarantees, i.e., the
behaviors promised by the system in the context of the environment.
To repair contracts, we can capture tradeoffs between assumptions and
guarantees in terms of minimization of a weighted norm of slacks. We
describe below our results in both the cases of non-adversarial and
adversarial environments.  
% In the former case, we propose repairs that restrict the set of environments compatible with the contract; in the latter, we look for a controller implementing the contract for any legal environments.  

\subsection{Non-Adversarial Environment}
\label{sec:nonadv}

% algorithm allows the designer to choose different weights for the slack variables of the environment assumptions and the system specifications, as changing the environment assumptions might be preferable.
For a contract, we make a distinction between controlled inputs $u_t$ and uncontrolled (environment) inputs $w_t$ of the dynamical system. In this section we assume that the environment signal $\wH$ can be predicted over a finite horizon and set to a known value for which the controller must be synthesized. With $\varphi \equiv \phie \rightarrow \phis$, equation~\eqref{eq:opt} reduces to: 
\begin{equation}
\begin{aligned}
	\label{eq:opt_min}
	 &\underset{\uH}{\text{minimize}}
	& \quad J(\xi^H(x_0, \uH, \wH))\\
	 &\text{subject to}
	& \quad \xi^H(x_0,\uH,\wH) \models \varphi,
\end{aligned}	
\end{equation}
%\pierluigi{State the original optimization problem.}
Because of the similarity of Problem~\eqref{eq:opt_min} and Problem~\eqref{eq:opt_simple}, we can then diagnose and repair a contract using the  methodology illustrated in Section~\ref{sec:diagnosis}. However, to reflect the different structure of the specification, i.e., its partition into assumption and guarantees, we adopt a weighted sum of the slack variables in Algorithm~\ref{algo:DiagnosisRepair}, allocating different weights to predicates in the assumption and guarantee formulas.
%However, we modify equation~\eqref{eq:slack2} by considering, instead of the $l_1$-norm, a weighted sum of the slack variables, allocating different weights to different parts of the specification. 
%  thus obtaining \ldots, where $\lambda_k$ be such weights for each slack variable, set by the designer:
%\begin{equation}
%\label{eq:slack_new}
%\begin{aligned}
%	 \underset{\s}{\text{minimize}}
%	&\quad \sum_{k=1}^{m_{ineq}+m_{eq}} \lambda_ks_k &\\
%	 \text{subject to}
%	&\quad f_i - s_i < 0 & \quad i \in \{ 1,\dotsc, m_{ineq}\}\\
%	&\quad f_j + s_j^p -s_j^n = 0 & \quad j \in \{ 1,\dotsc, m_{eq}\}
%\end{aligned}	
%\end{equation}
We can then provide the same guarantees as in Theorems~\ref{thm:SoundnessGeneral} and~\ref{thm:CompletenessGeneral}, where $\varphi \equiv \varphi_e \rightarrow \varphi_s$ and the minimality conditions are stated with respect to the weighted norm. 
%\pierluigi{Briefly restate the guarantees.} 
%\dorsa{problem 22 reduces  to problem 8, and by theorems 1 and 2 the approach is sound and complete, where $\varphi = \varphi_e' \rightarrow \varphi_s$, and the norm is a weighted norm, thus the minimality is with respect to the weighted norm.}

\begin{example}[Non-adversarial Race] \label{ex:sol2}
\sloppypar
We consider Example~\ref{sec:example2} with the same discretization step $\Delta t = 0.2$ and horizon $H=10$ as in Example~\ref{sec:example1}.
%  Our specification is of the form, $\psi := \psi_e \rightarrow \psi_s$. 
The MPC scheme results infeasible at time $0$. In fact, we observe that $\psi_e$ is always true as $v_{0}^{\rm adv} \geq 0.5$ and $\aadvt = 1 \geq 0$ holds at all times. Since $v_{0}^{ego} = 0$, the predicate $\psi_{s2} = \G_{[0, \infty)} (\vegot \geq 0.5)$ in $\psi_{s}$ is found to be failing. As in Section~\ref{sec:repair}, we can modify the conflicting predicates in the specification by using slack variables as follows: 
% The infeasibility occuring here is because although the assumption is satisfied, the system fails to satisfy its guarantees. To make the system feasible, we can either cause the failure of the assumptions or satisfaction of the guarantees.
$\vadvt + s_{e}(t) \geq 0.5$ (assumptions) and $\vegot + s_{s}(t) \geq 0.5$ (guarantees). However, we also assign a set of weights to the assumption ($\lambda_e$) and guarantee ($\lambda_s$) predicates, our objective being $\lambda_e |s_e| + \lambda_s |s_s|$. By setting $\lambda_s > \lambda_e$, 
we encourage modifications in the assumption predicate, thus obtaining $s_e = 0.06$ at time $0$ and zero otherwise, and $s_s = 0$ at all times.  
%Suppose we prefer tightening our assumption over the weakening our guarantee, we apply lower weights to the slacks in the predicates of the assumptions() since our objective is getting minimized. We obtain the slacks shown in Table~\ref{table:slack2}.
%\begin{table}
%\begin{center}
%\scalebox{0.9}{
%\begin{tabular}{ c|c|c|c|c|c|c|c|c|c|c}
%% \hline
% time &  0 & 0.2 & 0.4 & 0.6 & 0.8 & 1 & 1.2 & 1.4 & 1.6 & 1.8\\
%  \hline
%  $s_{e}$ & 0.06 & 0 & 0 & 0 & 0 & 0 & 0 & 0 & 0 & 0 \\
%%   \hline
%\end{tabular}}
%\end{center}
%\caption{\small{Slack variables for horizon, with $\Delta t = 0.2$, and $H=10$.}}
%\label{table:slack2}
%\end{table}
We can then set $\psi'_e =  (\vadv_0 \geq 0.56)$, which falsifies $\psi_e$ at time $0$, so that $\psi_e \rightarrow \psi_s$ is satisfied over the entire range.  Alternatively, by setting $\lambda_s < \lambda_e$,  we
% If we would prefer weakening our guarantees, we apply lower rights to the slacks in the predicates of the guarantees ($w_e > w_s$). We 
obtain the slack values in Table~\ref{table:slack3}, which lead to the following predicate repair: $\psi'_{s2} = \G_{[0, \infty)} (\vegot \geq -0.01)$. 
% Since $v_{0}^{ego} = 0$ and the controller can choose $\aegot \geq 0$, this is always satisfied.

\begin{table}[t]
\begin{center}
\scalebox{0.9}{
\begin{tabular}{ c|c|c|c|c|c|c|c|c|c|c}
% \hline
 time &  0 & 0.2 & 0.4 & 0.6 & 0.8 & 1 & 1.2 & 1.4 & 1.6 & 1.8\\
  \hline
  $s_{s}$ & 0.51 & 0.31 & 0.11 & 0 & 0 & 0 & 0 & 0 & 0 & 0 \\
%   \hline
\end{tabular}}
\end{center}
%\vspace{-10pt}
\caption{\small{Slack variables used in Example~\ref{sec:example2} and~\ref{ex:sol2}.}}
\label{table:slack3}
\end{table}

\begin{figure}[t]
\centering
%\vspace{-15pt}
\includegraphics[width=0.3\textwidth]{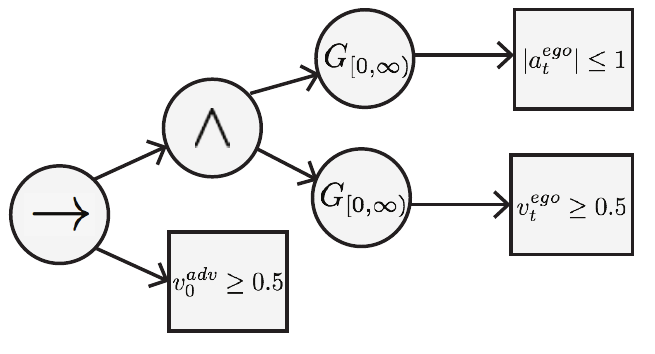}	
%\vspace{-10pt}
\caption{\small{Parse tree of $\psi \equiv \psi_e \rightarrow \psi_s$ used in Example~\ref{sec:example2} and~\ref{ex:sol2}.}
%\pierluigi{Any update needed for this figure? Can we make the text bigger?}
}
\label{fig:parseTree}
\end{figure}

We can also modify the time interval of the temporal operator associated with $\psi_{s2}$ to repair the overall specification. To do so, Algorithm~\ref{algo:DiagnosisRepair} uses the parse tree of $\psi_e \rightarrow \psi_s$ in Figure~\ref{fig:parseTree}. For any of the leaf node predicates $\mu_i$, $i\in \{1,2,3\}$,  
% ${\mu_1, \mu_2, \mu_3} = {(\vadvt \geq 0.5), (-1 \leq \aegot \leq 1), (\vegot \geq 0.5)}$. 
we get a support $\sigma_i = [0, 9]$, which is only limited by the finite horizon $H$.  
% the $\sigma^H_i = \sigma_i \wedge [0, 1.8]$, where $\sigma_i = [0, \infty]$ and $\sigma^H_i = [0, 1.8]$ for each $\mu_i$.
Then, based on the slack values in Table~\ref{table:slack3}, we can conclude $\sigma'_1 = \sigma'_2 = [0, 9]$ (the optimal slack values for these predicates are always zero), while
% As we have already proved, $\psi_e$ is always true, thus $\sigma'_1 = \sigma^H_1 = [0, 1.8]$. Since, $\psi_1$ is range of $\aegot$ over all times, we would prefer not changing the interval. From Table~\ref{table:slack3}, we see that 
$\sigma'_3 = [3, 9]$. For the given syntax tree, we also have $\sigma^{*}_1 = \sigma'_1$, $\sigma^{*}_2 = \sigma'_2$, and  $\sigma^{*}_3 = \sigma'_3$ for the temporal operator nodes that are parent nodes of $\mu_1$, $\mu_2$, and $\mu_3$, respectively. Since none of the above intervals is empty, a time interval repair is indeed possible 
% by only modifying $\sigma_3$, the support of predicate $\mu_3$, by restricting it to $\sigma^{*}_3$.
%
% We have $u^{*}_3 = u'_3 = u^H_3 = 1.8$, we can choose to keep it as $\infty$ as it doesn't change. Since, only $\sigma^{*}_3$ gets modified, we see that a modification in $\tau_3$ is possible. 
% From Fig.~\ref{fig:parseTree}, we can implement such a change 
by modifying the time interval of the parent node of $\mu_3$, thus achieving $\tau^*_3 = \sigma^*_3.$ This leads to the following proposed sub-formula $\psi'_{s2} = \G_{[0.6, \infty)} (\vegot \geq 0.5)$. In this example, repairing the specification over the first horizon is enough to guarantee controller realizability in the future. We can then keep the upper bound of the $\G$ operator to infinity.
% _s = \psi_1 \wedge \psi_2$ is a guarantee the system can provide given the assumption $\psi_e$ over the entire range.   
\end{example}

\subsection{Adversarial Environment}
\label{sec:adv}
\sloppypar
% In the previous sections, we only have discussed specifications, where all the elements  of system and environment are fully specified. 
When the environment can behave adversarially,  
% \ie environment has a choice and its actions are not fully defined. Such scenarios correspond to the nested min-max 
the control synthesis problem assumes the structure in~\eqref{eq:opt}. Specifically, in this paper,  
% while $\wH$ (acceleration of the environment vehicle) had a fixed value; however, now 
we allow $w_t$ to lie in  an interval $[w_{\text{min}}, w_{\text{max}}]$ at all times; this corresponds to the 
STL formula $\varphi_{w} = \G_{[0,\infty)} ({w}_{\text{min}} \leq w_t \leq {w}_{\text{max}})$.    
%$\wH$ to lie in an interval \ie 
%given some lower bound $\wlb$ and upper bound $\wub$ around the reference point $\wH_0$ \ie
% \begin{equation}
%$\wH \in [{\bf w}_{\text{min}} , {\bf w}_{\text{max}}].$ 
We decompose a specification $\varphi$ of the form $\varphi_w \land \varphi_e \rightarrow \varphi_s$, representing the contract, as $\varphi \equiv \varphi_{w} \rightarrow \psi$, where $\psi \equiv (\varphi_e \rightarrow \varphi_s)$. Our diagnosis and repair method is summarized in Algorithm~\ref{algo:DiagnosisRepairAdversarial}.

\begin{algorithm}[t]
\caption{\small{{\tt DiagnoseRepairAdversarial}}}\label{algo:DiagnosisRepairAdversarial}
\scriptsize
\begin{algorithmic}[1]
\medskip
\Procedure{{\tt DiagnoseRepairAdversarial}}{}
  \vspace{1mm}
  \State {\textbf{Input:} $\mathcal{P}$}
  \vspace{1mm}
  \State {\textbf{Output: } $\mathbf{u}^H$, $\mathcal{P}'$}
  \vspace{1mm}
  \State {$ (J,C) \gets$ {\tt GenMILP($\mathcal{P}$)}}
     \vspace{1mm}
%   \State{ $(\varphi_w \rightarrow (\varphi_e \rightarrow \varphi_s)) \gets \mathcal{P}.\varphi$}
%      \vspace{1mm}
  \State {$(\mathbf{u}_0^H, \wH_0, sat) \gets$ {\tt CheckSAT($J,C$)}}
  \vspace{1mm} 
  \If {$sat$} 
  \vspace{1mm}
  \State {$ \mathcal{W}^*_{cand} \gets$ {\tt SolveCEGIS($\uH_0, \mathcal{P}$)}}
  \vspace{1mm}
  \State {$ \mathcal{W}_{cand} \gets \mathcal{W}^*_{cand}$}
  \vspace{1mm}
  \While {$  \mathcal{W}_{cand} \neq \emptyset$}
    \vspace{1mm}
  \State {${\mathcal{P}}_w \gets$ {\tt RepairAdversarial($ \mathcal{W}_{cand}, \mathcal{P}$)}}
  \vspace{1mm}
   \State {$ \mathcal{W}_{cand} \gets$ {\tt SolveCEGIS($\uH_0, {\mathcal{P}}_w$)}}
   \vspace{1mm}
  \EndWhile
   \State {$ \mathcal{W}_{cand} \gets \mathcal{W}^*_{cand}$, $ \mathcal{P}_{\psi} \gets \mathcal{P}$}
  \vspace{1mm}
  \While {$  \mathcal{W}_{cand} \neq \emptyset$}
    \vspace{1mm}
  \State {${\mathcal{P}}_{\psi} \gets$ {\tt DiagnoseRepair($\mathcal{P}_{\psi}$)}}
  \vspace{1mm}
   \State {$ \mathcal{W}_{cand} \gets$ {\tt SolveCEGIS($\uH_0, {\mathcal{P}}_{\psi}$)}}
   \vspace{1mm}
  \EndWhile
  % \vspace{1mm}
  %    \Else {}
  %    \vspace{1mm}
  % \State{$\varphi' \gets \texttt{DiagnosisRepair}(\mathcal{P})$}
  %  \vspace{1mm}
  % \vspace{1mm}
  \State {$\mathcal{P}' \gets$ {\tt FindMin(${\mathcal{P}}_{w}, {\mathcal{P}}_{\psi}$)}}
  % \vspace{1mm}
   \EndIf
  % \vspace{1mm}
\EndProcedure 
\medskip
\end{algorithmic}
\end{algorithm}

We first check the satisfiability of the control synthesis problem by examining whether there exists a pair of $\uH$ and $\wH$ for which problem~\eqref{eq:opt} is feasible (\texttt{CheckSAT} routine):
\begin{equation}
\label{eq:opt_sat}
\begin{aligned}
	 &\underset{\uH,\wH}{\text{minimize}}
	& \quad J(\xi^H(x_0, \uH, \wH))\\
	 &\text{subject to}
	& \quad \xi^H(x_0,\uH, \wH) \models \varphi\\
	& & \quad \wH \models \varphi_w \land \varphi_e.\\
\end{aligned}	
\end{equation}
%\dorsa{$\varphi_e = (w_{min} \leq w \leq w_{max}) \wedge \phi'_e$, and after 25 how you update $w_{min}$ and $w_{max}$}
If problem~\eqref{eq:opt_sat} is unsatisfiable, we can use the techniques introduced in Section~\ref{sec:repair} and~\ref{sec:nonadv} to diagnose and repair the infeasibility. Therefore, in the following, we assume that \eqref{eq:opt_sat} is satisfiable, hence there exist $\uH_0$ and $\wH_0$ that solve~\eqref{eq:opt_sat}. 
% To check realizability, we use a procedure based on the \emph{Counter-Example Guided Inductive Synthesis} (CEGIS) algorithm in~\cite{Raman15} .
%\pierluigi{I think as a first step we consider this problem as a non-adversarial one and see if there is satisfiability. If it is not satisfiable then we can repair the specification as done in Section~\ref{sec:nonadv}. If the problem is instead satisfiable, we get an initial controller $\uH_0$ and environment $\wH_0$. We then address the problem of repairing to achieve ``realizability''. We use a procedure based on the \emph{Counter-Example Guided Inductive Synthesis} (CEGIS) algorithm in~\cite{Raman15}.}
%The first optimization finds $\uH$ for a fixed disturbance $\wH_0$ that is initialized at the first time step: 
%\begin{equation}
%\label{eq:opt_cegis1}
%\begin{aligned}
%	 &\underset{\uH}{\text{minimize}}
%	& \quad J(\xi^H(x_0, \uH, \wH_0))\\
%	 &\text{subject to}
%	& \quad \xi^H(x_0,\uH, \wH_0) \models \phis
%\end{aligned}	
%\end{equation}
%This optimization finds the best controller for the given $\wH_0$. Let $\uH_0$ denote this controller. In the next iteration, 
To check realizability, we use the following CEGIS loop (\texttt{SolveCEGIS} routine). By first fixing the control trajectory to $\uH_0$, we find the worst case disturbance trajectory $\wH_1$ that minimizes the robustness value of $\varphi$ by solving the following problem:
 %We then update $\wH_0 \leftarrow \wH_1$, and repeat iterating between the two optimizations in equations~\eqref{eq:opt_cegis1} and~\eqref{eq:opt_cegis2} as long as the robustness value is above $\rho_{min}$.
\begin{equation}
\label{eq:opt_cegis2}
\begin{aligned}
	 &\underset{\wH}{\text{minimize}}
	& \quad \rho^\varphi(\xi^H(x_0, \uH, \wH),0)\\
	 &\text{subject to}
% 	& \quad \xi^H(x_0,\uH_0, \wH) \models \phie \wedge \varphi_w
	& \quad  \wH \models \phie \wedge \varphi_w
\end{aligned}	
\end{equation}
with  $\uH = \uH_0$. 
%\pierluigi{Need to revise the above problem. Is the robustness in 0? compare with the one in the the CEGIS paper}
%
% Similar to examples in the previous sections, we provide feedback about the controller if the first optimization~\eqref{eq:opt_cegis1} fails.
% We include slack variables to the constraints of equation~\eqref{eq:opt_cegis1}, and construct repair mechanisms based on the scale and timing information of the slack variables.
% However, if equation~\eqref{eq:opt_cegis1} is feasible, there exists a $\uH_0$ that satisfies all the constraints for a fixed $\wH_0$.
% While~\eqref{eq:opt_cegis2} is always feasible as long as the set of allowed environment trajectories is nonempty, 
% will never become infeasible, since the trajectory only needs to satisfy the constraints, while minimizing the robustness function. 
The optimal $\wH_1$ from~\eqref{eq:opt_cegis2} will falsify the specification if the resulting robustness value is below zero\footnote{\small{A tolerance $\rho_{min}$ can be selected to accommodate approximation errors, i.e., $\rho^\varphi(\xi^H(x_0, \uH_0, \wH_1),0) < \rho_{min}$.}}.
% i.e., $\rho^\varphi(\xi^H(x_0, \uH_0, \wH_1),0) < 0$.
% Such infeasibilities appear because of the value of disturbance $\wH_1$, and we would like to detect and repair them by adding slack variables to constrain the range of $\wH$.
If this is the case, we look for a $\uH_1$ which solves \eqref{eq:opt_min} with the additional restriction of $\wH \in {\cal W}_{cand} = \{\wH_1\}$. If this step is feasible, we once again attempt to find a worst-case disturbance sequence $\wH_2$ that solves~\eqref{eq:opt_cegis2} with $\uH = \uH_1$: this is the counterexample-guided inductive step. At each iteration $i$ of this CEGIS loop, the set of candidate disturbance sequences ${\cal W}_{cand}$ expands to include $\wH_i$. If the loop terminates at iteration $i$ with a successful $\uH_i$ (one for which the worst case disturbance $\wH_i$ in \eqref{eq:opt_cegis2} has positive robustness), we conclude that the formula $\varphi$ is realizable. 

The CEGIS loop may not terminate if the set ${\cal W}_{cand}$ is infinite. We, therefore, run it for a maximum number of iterations. 
% before declaring failure. 
If \texttt{SolveCEGIS} fails to find a controller sequence prior to the timeout, then \eqref{eq:opt_min} is infeasible for the current ${\cal W}_{cand}$,  i.e., there is no control input that can satisfy $\varphi$ for all disturbances in ${\cal W}_{cand}$. We conclude that the specification is not realizable (or, equivalently, the contract is inconsistent). 
% because of the boundary conditions on $w_t$. 
While this infeasibility can be repaired by modifying $\psi$ based on the techniques in Section~\ref{sec:repair} and~\ref{sec:nonadv}, an alternative solution 
% Algorithm~\ref{algo:DiagnosisRepairAdversarial} aims to 
is to repair $\varphi_w$ by minimally pruning the bounds on $w_t$ (\texttt{RepairAdversarial} routine). 

To do so, given a small tolerance $\epsilon \in \mathbb{R}^+$, we find 
\begin{equation} \label{eq:opt_real}
% \begin{aligned}
	w_u = \max_{\substack{\wH_i \in {\cal W}_{cand} \\ t \in \{1,\ldots,H-1\}}} w_{i,t} \qquad
	w_l = \min_{\substack{\wH_i \in {\cal W}_{cand} \\ t \in \{1,\ldots,H-1\}}} w_{i,t}
 % \end{aligned}
\end{equation}
and define $s_u = w_{\max} - w_u$ and  $s_l = w_l - w_{\min}$. We then use $s_u$ and $s_l$ to update the range for $w_t$  in $\varphi_w$ to a maximal interval $[w'_{\min}, w'_{\max}] \subseteq [w_{\min}, w_{\max}]$ and such that at least one $\wH_i \in {\cal W}_{cand}$ is excluded. Specifically, if $s_u \leq s_l$, we set $[w'_{\min}, w'_{\max}] = [w_{\min}, w_u - \epsilon]$; otherwise we set $[w'_{\min}, w'_{\max}] = [w_l + \epsilon, w_{\max}]$. The smaller the value of $\epsilon$, the larger the resulting interval. 
%
%To do so, we propose to solve the following problem, where slack variables are added to constraints encoding the equality of $\wH$ and each $\wH_i \in {\cal W}_{cand}$. 
%%\begin{equation}
%%\label{eq:slack3}
%%\begin{aligned}
%%	\underset{\s}{\text{minimize}}
%%	& \quad \text{norm}(\s) &\\
%%	 \text{subject to}
%%	& \quad \fdyn_i \leq 0 & \quad i \in \{ 1,\dotsc, m_d\}\\
%%	%&\quad \fbound_j \leq 0 & \quad j \in \{ 1,\dotsc, m_i\}\\
%%	& \quad \fstl_j \leq 0 & \quad j \in \{ 1,\dotsc, m_s\}\\
%%	& \quad \wH + s_l^p - s_l^n = \wH_1 & \quad l \in \{ 0,\dotsc, H-1\}\\
%%\end{aligned}	
%%\end{equation}
%\begin{equation}
%	\label{eq:opt_real}
%\begin{aligned}
%	&\underset{\s,\uH}{\text{minimize}}
%	& \quad ||\s||+J(\xi^H(x_0, \uH, \wH))&\\
%	&\text{subject to}
%	& \quad \xi^H(x_0,\uH, \wH) \models \psi  &\\
%	% & \quad & \wH + s_l^p - s_l^n = \wH_1  \quad l \in \{ 0,\dotsc, H-1\}
%	% & \quad & \wH \models \varphi_e &\\
%	& \quad \wH + \s_i^p - \s_i^n \geq \wH_i & \quad \wH_i \in {\cal W}_{cand}\\
%	& \quad \wH + \s_i^p - \s_i^n \leq \wH_i \\%& \quad l \in \{ 0,\dotsc, H-1\}&\\	
%	& \quad \s_i^p \geq 0 \quad \s_i^n \geq 0 %& \quad \wH_i \in {\cal W}_{cand}
%\end{aligned}
%\end{equation}
%Problem~\eqref{eq:opt_real} aims to minimize the norm of the slack variables and the objective function in problem~\eqref{eq:opt_sat} while enforcing satisfaction of $\psi$. Here $\s$ is the vector of all the slack variables, including $\s_i^p$ and $\s_i^n$ for each $\wH_i$. We use $\s$ to update the range for $w_t$  in $\varphi_w$ to an interval $[w'_{min}, w'_{max}] \subseteq [w_{min}, w_{max}]$. 
%
Finally, we use the updated formula $\varphi'_w$ to run \texttt{SolveCEGIS} again until a realizable control sequence $\mathbf{u}^{H}$ is found. 
% 
% The second while loop in Algorithm~\ref{algo:DiagnosisRepairAdversarial}, lines $11$ to $14$, implements the CEGIS loop while calling Algorithm~\ref{algo:DiagnosisRepair} to repair $\psi$. 
In Algorithm~\ref{algo:DiagnosisRepairAdversarial}, assuming a predicate repair procedure, \texttt{FindMin} provides the solution with minimum slack norm between the ones repairing $\psi$ and $\varphi_w$.
\begin{example}[Adversarial Race]
\sloppypar
We consider the specification in Example~\ref{sec:example3}. 
% and decompose the contract guarantees as $\phi := \phi_w \rightarrow (\phi_e \rightarrow \phi_s)$. 
% The system assumes that the environment maintains its acceleration in the given range, $\phi_w$. The specification can be made feasible by relaxing the guarantees $\phi_s$, the repair for this case is already discussed in Section~\ref{sec:repair}.
%Consider the STL formula $\varphi_{3e} \rightarrow \varphi_{3s}$ for our driving scenario, with $\varphi_{3e}$ and $\varphi_{3s}$ of the following form:
%
%\begin{equation}
%\begin{aligned}
%&\varphi_{3e} = \G_{[0,\infty)} \big ( 1\leq w_t \leq 3 \big )\\
%	&\varphi_{3s} =  \G_{[0,\infty)} \big( (\vadvt < 2.5) \wedge (\vegot < 0.5)  \land (|u_t| \leq 1) \big ).\\
%\end{aligned}	
%\end{equation}
%\begin{equation}
%\begin{aligned}
%	&\varphi_3  =  \G_{[0,\infty)} \big( (\vadvt < 2.5) \wedge (\vegot < 0.5) \big )\\
%		&\psi_3 = \G_{[0,\infty)} \big ( |u_t| \leq 1 \big )
%\end{aligned}	
%\end{equation}
%We require that, at all times, both the velocity of the adversary and \emph{ego} vehicles be below a given threshold, and the control energy be bounded, under the assumption that the environment varies within the interval $[1,3]$ (i.e., $w_0 = 2$, $w_{ub} = - w_{lb} = 1$ at all times).
%\pierluigi{I have the same problem with the interval, since the velocity will anyway keep growing if the acceleration is positive}. 
For the same horizon as in the previous examples, 
% we let $\wH_0 = [2,2,2,2,2,2]$, and we let each element range between $[2-1, 2+1] = [1,3]$. 
after solving the satisfiability problem, for the fixed $\uH_0$, the CEGIS loop returns $\aadvt = 2$ for all $t \in \{0,\dotsc,H-1\}$ as the single 
element in ${\cal W}_{cand}$ for which no controller sequence can be found. 
% which causes the robustness value $\rho_{\text{min}}$ to fall below $0$.
%  threshold \ie $\rho(\xi^H(x_0,\uH_0,\wH_1))<\rho_{min}$.
% Since this infeasibility is caused by the disturbance value $\wH_1$, 
We then choose to tighten the environment assumptions to make the controller realizable, by shrinking the bounds on $\aadvt$ by using 
Algorithm~\ref{algo:DiagnosisRepairAdversarial} with $\epsilon = 0.01$. After a few iterations, we finally obtain $w'_{\min} = 0$ and $w'_{\max}=1.24$, and therefore  
%\begin{equation}
	$\phi'_w = \G_{[0,\infty)} \big( 0 \leq \aadvt \leq 1.22 \big).$
%\end{equation} 
\end{example}

Under the assumption that \texttt{SolveCEGIS} terminates before reaching the maximum number of iterations\footnote{If this is not the case, then Algorithm~\ref{algo:DiagnosisRepairAdversarial} terminates with \texttt{UNKNOWN}.}, and within the selected tolerance $\epsilon$, the following theorems state the properties of Algorithm~\ref{algo:DiagnosisRepairAdversarial}. 
%\pierluigi{introduce $\epsilon$ in the theorem statements?}

\begin{theorem}[Soundness] \label{thm:SoundnessContract}
Given a controller synthesis problem $\P = (f_d, g_d, x_0, \varphi, J)$, such that~\eqref{eq:opt} is infeasible at time $t$, let $\varphi' \in \texttt{REPAIR}_{\mathcal{D},\mathcal{T}}(\varphi)$ be the repaired formula returned from Algorithm~\ref{algo:DiagnosisRepairAdversarial} without human intervention,
% ~\ie $w={\bf 1}$, 
for a given set of predicates $\mathcal{D}$ or time interval $\mathcal{T}$. Then, $\P' = (f_d, g_d, x_0, \varphi', J)$ is feasible at time $t$ and ($\varphi'$, $\mathcal{D}$, $\mathcal{T}$) satisfy the minimality conditions in Problem~\ref{prob:repair-contract}.
%Under the assumption that all $w=1$:
%\begin{itemize}
%\item  $\varphi'$ is feasible.	
%\item  Under the assumption that $\varphi'$ is obtained only by predicate repair, and $\mathcal{D}=\{\mu_1,\dotsc,\mu_d\}$ is the set of repaired predicates, we let $s_i = \mu_i^* - \mu_i$ for $i \in \{1,\dotsc,d\}$ and $s_{\mathcal{D}} = (s_1,\dotsc,s_d)$. Then, 
%\begin{equation}
% \nexists \: (\mathcal{D'}, s_{\mathcal{D'}}) \quad \emph{s.t.} \quad ||s_{\mathcal{D'}}|| \leq ||s_{\mathcal{D}}||.
%\end{equation}
%\item Under the assumption that $\varphi'$ is obtained only by time interval repair, and $\mathcal{T} = \{ \tau_1,\dotsc, \tau_d \}$ be the set of time intervals associated to $\mathcal{D}$, we let $\mathcal{T}^* = \{\tau_1^*,\dotsc, \tau_d^*\}$ be the repaired intervals. Then,
%\begin{equation}
%\begin{aligned}
%	\nexists \: \mathcal{T'} = \{\tau_1',\dotsc, \tau_d'\},  \quad \emph{s.t.} \quad ||\tau^*_i|| \leq || \tau'_i||, \\
%	\quad \emph{for any} \quad i \in \{1,\dotsc,d \}.
%\end{aligned}
%\end{equation}
%\end{itemize}
\end{theorem}

%\ignore{
\begin{proof}[(Theorem~\ref{thm:SoundnessContract})]
We recall that $\varphi \equiv \varphi_w \rightarrow \psi$. Moreover, Algorithm~\ref{algo:DiagnosisRepairAdversarial} provides the solution with minimum slack norm between the ones repairing $\psi$ and $\phi_w$ in the case of predicate repair. Then, when $\psi = \varphi_e \rightarrow \varphi_s$ is modified using Algorithm~\ref{algo:DiagnosisRepair}, soundness is guaranteed by Theorem~\ref{thm:SoundnessGeneral} and the termination of the CEGIS loop. On the other hand, assume Algorithm~\ref{algo:DiagnosisRepairAdversarial} modifies the atomic predicates in 
$\phi_w$. Then, the \texttt{RepairArdversarial} routine and~\eqref{eq:opt_real}, together with the termination of the CEGIS loop, assure that $\varphi_w$ is also repaired in such a way that the controller is realizable, and the length of the bounding box around $w_t$ is maximal within an error bounded by $\epsilon$ (i.e., it differs from the maximal interval length by at most $\epsilon$), which concludes our proof.   
%We recall that $\varphi \equiv \varphi_w \rightarrow \psi$. Then, if Algorithm~\ref{algo:DiagnosisRepairAdversarial} chooses to 
%modify $\psi = \varphi_e \rightarrow \varphi_s$ using Algorithm~\ref{algo:DiagnosisRepair}, soundness is guaranteed by Theorem~\ref{thm:SoundnessGeneral} and the termination of the CEGIS loop. On the other hand, if Algorithm~\ref{algo:DiagnosisRepairAdversarial} opts for modifying the atomic predicates in 
%$\phi_w$, the solution of problem~\eqref{eq:opt_real}, together with the termination of the CEGIS loop, assures that $\varphi_w$ is also repaired in such a way that the controller is realizable, and any reduction in the bounding box around $w_t$ is minimal, hence the range of allowed values is maximal.  
%
% If we get $\varphi'$ then $\mathcal{P}'$ with $\varphi'$ is feasible and minimal.\\
% 1) If algorithm 4 decides to change $\psi$ , then because we use algorithm 1 and because of the CEGIS loop the results of it are sound and complete.\\
% 2) If we change  by (26) and CEGIS we get feasibility and minimality.
% The predicate slack is a non-trivial function of $S^*$ in (26), but the effect is that $[w_{min}, w_{max}]$ is maximized, so the correction is minimized.
\end{proof}
%}

\begin{theorem}[Completeness] \label{thm:CompletenessContract}
\sloppypar
Assume the controller synthesis problem $\P = (f_d, g_d, x_0, \varphi, J)$ results in~\eqref{eq:opt} being infeasible at time $t$.
If there exist a set of predicates $\mathcal{D}$ and time-intervals $\mathcal{T}$ such that there exists $\Phi \subseteq \texttt{REPAIR}_{\mathcal{D},\mathcal{T}}(\varphi)$ for which $\forall \ \phi \in \Phi$, $\P' = (f_d, g_d, x_0, \phi, J)$  is feasible at time $t$ and ($\phi$, $\mathcal{D}$, $\mathcal{T}$) are minimal in the sense of Problem~\ref{prob:repair-contract}, then Algorithm~\ref{algo:DiagnosisRepairAdversarial} returns a repaired formula $\varphi'$ in $\Phi$.  
\end{theorem}

%\ignore{
\begin{proof}[(Theorem~\ref{thm:CompletenessContract})]
%As discussed in the proof of Theorem~\ref{thm:SoundnessContract}, if Algorithm~\ref{algo:DiagnosisRepairAdversarial} chooses to 
%modify $\psi = \varphi_e \rightarrow \varphi_s$ using Algorithm~\ref{algo:DiagnosisRepair}, completeness is guaranteed by Theorem~~\ref{thm:CompletenessGeneral} and the termination of the CEGIS loop.
%On the other hand, let us assume there exists a minimum norm repair for the atomic predicates of $\varphi_w$, which returns a maximal interval $[w'_{\min}, w'_{\max}] \subseteq [w_{min}, w_{max}]$. Then, problem~\eqref{eq:opt_real} would be feasible and, given the termination of the CEGIS loop, would also produce a predicate repair, and a corresponding interval $[w''_{min}, w''_{max}]$, which make the control synthesis realizable. By minimization of the norm in~\eqref{eq:opt_real}, we can also conclude that  $[w''_{min}, w''_{max}]$ is maximal (in the sense of norm minimization), hence $\varphi' \in \Phi$ holds. 
% 
As discussed in the proof of Theorem~\ref{thm:SoundnessContract}, if Algorithm~\ref{algo:DiagnosisRepairAdversarial}  
modifies $\psi = \varphi_e \rightarrow \varphi_s$ using Algorithm~\ref{algo:DiagnosisRepair}, completeness is guaranteed by Theorem~~\ref{thm:CompletenessGeneral} and the termination of the CEGIS loop.
On the other hand, let us assume there exists a minimum norm repair for the atomic predicates of $\varphi_w$, which returns a maximal interval $[w'_{\min}, w'_{\max}] \subseteq [w_{\min}, w_{\max}]$. Then, given the termination of the CEGIS loop, by repeatedly applying \eqref{eq:opt_real} and \texttt{RepairAdversarial}, it is also possible to produce a predicate repair such that the corresponding interval $[w''_{\min}, w''_{\max}]$ makes the control synthesis realizable and is maximal within an error bounded by $\epsilon$ (i.e., its length differs by at most $\epsilon$ from the one of the maximal interval $[w'_{\min}, w'_{\max}]$). Hence, $\varphi' \in \Phi$ holds.
\end{proof}
%}

\ignore{
\begin{proof}[Sketch]
We recall that $\varphi \equiv \varphi_w \rightarrow \psi$. Then, both the results follow from Theorem~\ref{thm:SoundnessGeneral} and~\ref{thm:CompletenessGeneral} as well as the minimization problem~\eqref{eq:opt_real}. In fact, the minimality conditions for a repair on $\psi = \varphi_e \rightarrow \varphi_s$ using Algorithm~\ref{algo:DiagnosisRepair} are guaranteed by Theorem~\ref{thm:SoundnessGeneral} and~\ref{thm:CompletenessGeneral}. On the other hand, problem~\eqref{eq:opt_real} assures that $\varphi_w$ is also repaired in such a way that any reduction in the bounding box around $w_t$ is minimal, hence the range of allowed values is maximal.  
% If Equation~\eqref{eq:opt_real} returns a pruned interval, $\tau^{*}_w = [w^{*}_{min}, w^{*}_{max}]$, then there $\nexists \tau'_w= [w'_{min}, w'_{max}]$ such that $\tau^{*}_w \subseteq \tau'_w$. Hence, $\tau^{*}_w$ is a maximal set. 
%$\varphi = \varphi_w \rightarrow \psi$ can be re-written as $\varphi = \neg(\varphi_w) \vee \psi$. The minimality on $\varphi$ can be translated to a minimality of $\neg(\varphi_w)$ and minimality of $\psi$. Minimality of $\neg(\varphi_w)$ can be seen as a maximality of $\varphi_w$. From the above statements both the maximality of $\varphi_w$ and minimality of $\psi$ is assured. Hence, the minimality of $\varphi'$ from Algorithm~\eqref{algo:DisagnosisRepairAdversarial} holds. 
\end{proof}
}
%%% Local Variables: 
%%% mode: latex
%%% TeX-master: "root"
%%% End: 

\section{Case Studies}
\label{sec:examples}
\subsection{Autonomous Driving}
\label{sec:autonomouDriving}
We consider the problem of synthesizing a controller for an autonomous vehicle in a city driving scenario. We analyze the following two tasks:
(i) changing lanes on a busy road; 
(ii) performing an unprotected left turn at a signalized intersection.	
We use a simple point-mass model for the vehicles on the road.  
For each vehicle, we define the state as $\x = [x \enskip y \enskip \theta \enskip v]^{\top}$, where $x$ and $y$ denote the coordinates, and $\theta$ and $v$ represent the direction and speed, respectively.
Let $\uu = [u_1 \enskip u_2]^{\top}$ be the control input for each vehicle, where $u_1$ is the steering input and $u_2$ is the acceleration. Then, the vehicle's state evolves according to the following dynamics:
\begin{equation}
\openup -1\jot
\begin{aligned}
	&\dot{x} = v \cos\theta\\
	&\dot{y} = v \sin\theta \\
	&\dot{\theta} = {v \cdot u_1}/{m}\\
	&\dot{v} = u_2,
\end{aligned}	
\end{equation}
where $m$ is the vehicle mass. To determine the control strategy, we linearize the overall system dynamics around the initial state at each run of the MPC, which is completed in less than 2~s on a 2.3-GHz Intel Core i7 processor with 16-GB memory.
%
% \dorsa{Note that diagnosis and repair are not implemented in real time at the moment, and they are regarded as tools to help debug the specification at the design time.}
%  
We further impose the following constraints on the \emph{ego} vehicle (i.e., the vehicle under control): (i) a minimum distance must be established  between the \emph{ego} vehicle and other cars on the road to avoid collisions; (ii) the \emph{ego} vehicle must obey the traffic lights;
(iii) the \emph{ego} vehicle must stay within its road boundaries.

\subsubsection{Lane Change}
We consider a lane change scenario on a busy road as shown in Fig.~\ref{fig:MergeLeftInfeas}. The \emph{ego} vehicle is in red. 
% while the yellow cars are the other vehicles on the road.
 \emph{Car 1} is at the back of the left lane, \emph{Car 2} is in the front of the left lane, while \emph{Car 3} is on the right lane.
The states of the vehicles are initialized as follows: 
 $x_0^{\text{Car 1}} = [-0.2\enskip -1.5 \enskip \frac{\pi}{2} \enskip 0.5]^\top$, $x_0^{\text{Car 2}} = [-0.2 \enskip 1.5 \enskip \frac{\pi}{2} \enskip 0.5]^\top$, $x_0^{\text{Car 3}} = [0.2 \enskip 1.5 \enskip \frac{\pi}{2} \enskip 0]^\top$, and $x_0^{\text{ego}} = [0.2 \enskip -0.7 \enskip \frac{\pi}{2}\enskip 0]^\top$. The control inputs for \emph{ego} and \emph{Car 3} are initialized at $[0 \enskip 0]^\top$; the ones for \emph{Car 1} and \emph{Car 2} are set to $u_0^{\text{Car 1}} = [0 \enskip1]^\top$ and $u_0^{\text{Car 2}} = [0\enskip -0.25]^\top$.
The objective of \emph{ego} is to safely change lane, while satisfying the following requirements:
\begin{equation}
\label{eq:CarSTL}
\begin{aligned}
	&\varphi_{\text{str}} = \G_{[0,\infty)}(|u_{1}| \leq 2)& \quad \text{\small{Steering Bounds}}\\
	&\varphi_{\text{acc}}= \G_{[0,\infty)}(|u_{2}| \leq 1)&\quad \text{\small{Acceleration Bounds}} \\
	&\varphi_{\text{vel}} = \G_{[0,\infty)}(|v| \leq 1)& \quad \text{\small{Velocity Bounds}}
\end{aligned}	
\end{equation}
The solid blue line in Fig.~\ref{fig:MergeLeft} is the trajectory of \emph{ego} as obtained from our MPC scheme, while the dotted green line is the future trajectory pre-computed for a given horizon at a given time. MPC becomes infeasible at time $t=1.2$~s when the no-collision requirement is violated, and a possible collision is detected between the \emph{ego} vehicle and \emph{Car 1} before the lane change is completed (Fig.~\ref{fig:MergeLeft}a). Our solver takes $2$~s, out of which $1.4$~s are needed to generate all the IISs, consisting of $39$ constraints. To make the system feasible, the proposed repair increases both the acceleration bounds and the velocity bounds on the \emph{ego} vehicle as follows:
\begin{equation}
\begin{aligned}
 \varphi_{\text{acc}}^{\text{new}} = \G_{[0,\infty)}(|u_{2}| \leq 3.5) \\
 \varphi_{\text{vel}}^{\text{new}} = \G_{[0,\infty)}(|v| \leq 1.54)	
\end{aligned}
\end{equation}
When replacing the initial requirements $\varphi_{\text{acc}}$ and $\varphi_{\text{vel}}$ with the modified ones, the revised MPC scheme  allows the vehicle to travel faster and safely complete a lane change maneuver, without risks of collision, as shown in Fig.~\ref{fig:MergeLeftFeas}. 

\begin{figure}
\centering
\begin{subfigure}[b]{0.22\textwidth}
\centering
	\includegraphics[width = \textwidth]{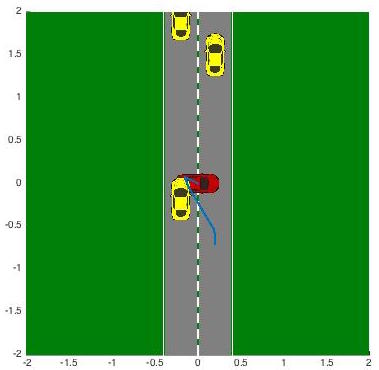}
	\caption{}
	\label{fig:MergeLeftInfeas}
\end{subfigure}%
\hspace{12pt}
\begin{subfigure}[b]{0.22\textwidth}
\centering
	\includegraphics[width = \textwidth]{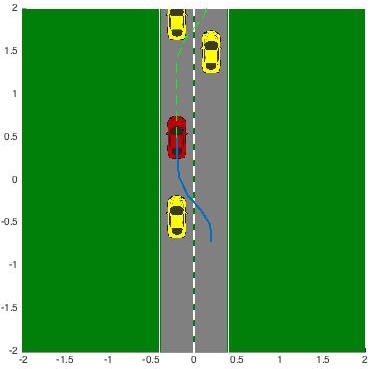}
	\caption{}
	\label{fig:MergeLeftFeas}
\end{subfigure}
%\vspace{-5pt}
\caption{\small{Changing lane is infeasible at $t=1.2$~s in (a) and gets repaired in (b).}}
\label{fig:MergeLeft}
\end{figure}

\subsubsection{Unprotected Left Turn}
In the second scenario, we would like the \emph{ego} vehicle to perform an unprotected left turn at a signalized intersection, where the \emph{ego} vehicle has a green light and is supposed to yield to oncoming traffic, represented by the yellow cars crossing the intersection in Fig.~\ref{fig:LeftTurn}. 
% where the ego vehicle (red car) has to make a left turn at the intersection, while there are two oncoming environment cars () .
The environment vehicles are initialized at the states $x_0^{\text{Car 1}} = [-0.2\enskip 0.7 \enskip -\frac{\pi}{2} \enskip 0.5]^\top$ and $x_0^{\text{Car 2}} = [-0.2 \enskip 1.5 \enskip -\frac{\pi}{2} \enskip 0.5]^\top$, while the \emph{ego} vehicle is initialized at $x_0^{\text{ego}} = [0.2 \enskip -0.7 \enskip \frac{\pi}{2}\enskip 0]^\top$. The control input for each vehicle is initialized at $[0 \enskip 0]^\top$. Moreover, we use the same bounds as in~\eqref{eq:CarSTL}.
\begin{figure}
\centering
\begin{subfigure}[b]{0.22\textwidth}
\centering
	\includegraphics[width = \textwidth]{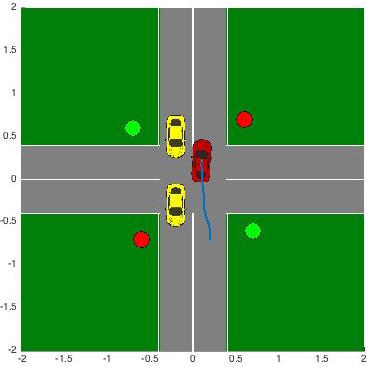}
	\caption{}
	\label{fig:LeftTurnInfeas}
\end{subfigure}%
\hspace{10pt}
\begin{subfigure}[b]{0.22\textwidth}
\centering
	\includegraphics[width = \textwidth]{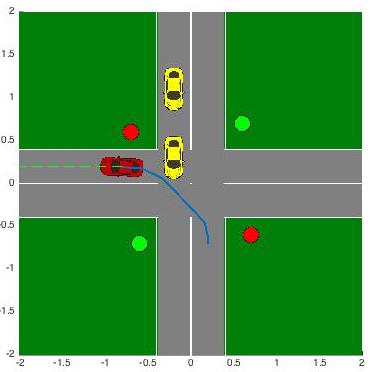}
	\caption{}
	\label{fig:LeftTurnFeas}
\end{subfigure}
%\vspace{-5pt}
\caption{\small{Left turn becomes infeasible at time $t=2.1$ s in (a) and is repaired in (b).}}
\label{fig:LeftTurn}
\end{figure}

The MPC scheme becomes infeasible at $t=2.1$~s. The solver takes $5$~s, out of which $2.2$~s are used to generate the IISs, including 56 constraints. As shown in Fig.~\ref{fig:LeftTurnInfeas}, the \emph{ego} vehicle yields in the middle of intersection for the oncoming traffic to pass. However, the traffic signal turns red in the meanwhile, and there is no feasible control input for the \emph{ego} vehicle without breaking the traffic light rules. 
% Furthermore, the three boundary conditions in equation~\eqref{eq:CarSTL} appear in the subset of infeasible constraints. 
Since we do not allow modifications to the traffic light rules, the original specification is repaired again by increasing the bounds on  acceleration and velocity, thus obtaining:
%
%\vspace{-10pt}
\begin{equation}
\begin{aligned}
	\varphi^{\text{new}}_{\text{acc}}  &=&  \G_{[0,\infty)}\big (|u_{2}| \leq 11.903 \big )\\	
	\varphi^{\text{new}}_{\text{vel}}  &=&  \G_{[0,\infty)}\big(|v| \leq 2.42\big )
\end{aligned}	
\end{equation}
%\vspace{0pt}
%After this repair, another infeasibility appears at $t=2.$~s due to locally linearizing the nonlinear dynamics. As a result of the local linearization, the diagnosis might not be able to guarantee feasibility of the system after the repairs. This infeasibility can get fixed by adding more slack variables to $\varphi_{\text{vel}}$:
%\begin{equation}
%\varphi^{\text{new}}_{\text{vel}} = \G_{[0,\infty)}(|v| \leq 2.42)
%\end{equation}
As shown by the trajectory in Fig.~\ref{fig:LeftTurnFeas}, under the assumptions and initial conditions of our scenario, higher allowed velocity and acceleration make the \emph{ego} vehicle turn before the oncoming cars get close or cross the intersection.

%\pierluigi{Would like to discuss further the relationship between MILP generation time, MILP solution, and runtime of Algorithm 1}

%\vspace{-5pt}
%\subsection{Quadrotor Control}
%\label{sec:quadrotor}
%\input{quadrotor.tex}

\subsection{Aircraft Electric Power System}

%\begin{figure}[t]
%	\centering
%	%\vspace{-15pt}
%	\includegraphics[scale = 0.7]{figures/EPS.pdf}
%	%\vspace{-10pt}
%	\caption{\small{Simplified model of an aircraft electric power system.}}
%	\label{fig:EPSModel}
%\end{figure}
%
%\begin{figure}[t]
%	%\vspace{-20pt}
%	\centering
%	\includegraphics[scale = 0.55]{figures/EPSplot.pdf}
%	\vspace{-10pt}
%	\caption{\small{A counterexample trajectory for the electric power system (380 ms run). The blue, green and red lines represent environment, state, and controller variables, respectively.}}
%	\label{fig:EPS}
%\end{figure}

\begin{figure}[t]
	% \vspace{-10pt}
	 \centering
	\begin{tabular}{cc}
	\includegraphics[scale = 0.6, angle = 90]{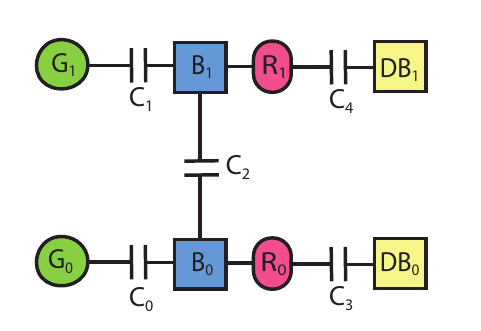} &
	\includegraphics[scale = 0.6]{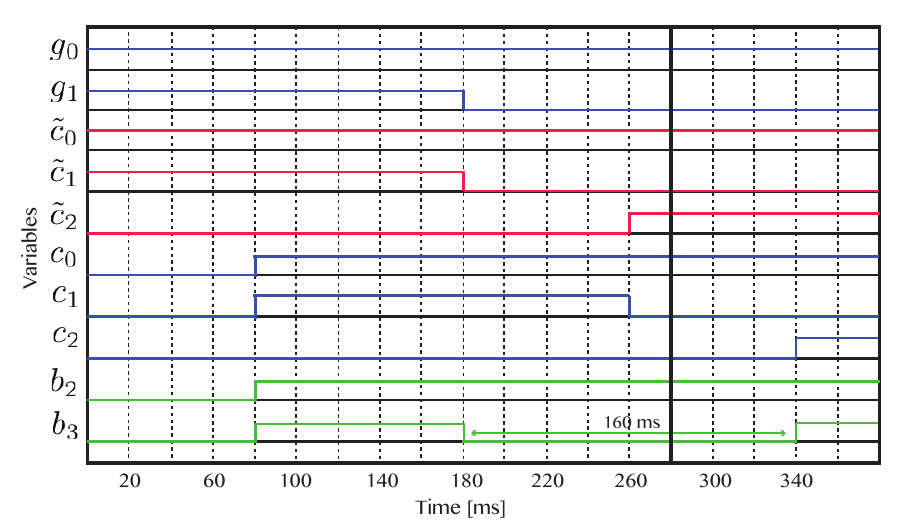}
	\end{tabular}
	%\vspace{-10pt}
	\caption{\small{Simplified model of an aircraft electric power system (left) and counterexample trajectory (right). The blue, green and red lines represent environment, state, and controller variables, respectively, for a  380-ms run.}}
	\label{fig:EPS}
\end{figure}

%We apply our algorithms to the diagnosis and repair of requirements for the power distribution in a passenger aircraft. 
Fig.~\ref{fig:EPS} shows a simplified architecture for the primary power distribution system in a passenger aircraft~\cite{Nuzzo14}.
% ~\cite{Michalko2008}. 
Two power sources, the left and right generators $G_0$ and $G_1$, deliver power to a set of high-voltage AC and DC buses ($B_0$, $B_1$, $DB_0$, and $DB_1$) and their loads. 
% (not shown in Fig.~\ref{fig:EPS}). 
AC power from the generators is converted to DC power by rectifier units ($R_1$ and $R_2$). A bus power control unit (controller) monitors the availability of power sources and configures a set of electromechanical switches, denoted as contactors ($C_0, \ldots, C_4$), such that essential buses remain powered even in the presence of failures, while satisfying a set of safety, reliability, and real-time performance requirements~\cite{Nuzzo14}. Specifically, we assume that only the right DC bus $DB_1$ is essential, and use our algorithms to check the feasibility of a controller that accommodates a failure in the right generator $G_1$, by rerouting power from the left generator to the right DC bus in a time interval which is less than or equal to $t_{\text{max}}=100$~ms. In addition, the controller must satisfy the following set of requirements, all captured by an STL contract.

\textbf{Assumptions.} \emph{When a contactor receives an open (close) signal, it shall become open (closed) in 80~ms or less}. Let the time discretization step $\Delta t = 20$~ms, $\tilde{c}_i$, $i\in\{0,\ldots, 4\}$ be a set of Boolean variables describing the controller signal (where $1$ stands for ``closed'' and $0$ for ``open"), $c_i$, $i\in\{0,\ldots, 4\}$ be a set of Boolean variables denoting the state (open/closed) of the contactors. We can capture the system assumptions via a conjunction of formulas of the form: $\G_{[0, \infty)}(\tilde{c}_i \rightarrow \F_{[0,4]}c_i)$, providing a model for the discrete-time binary-valued contactor states. The actual delay of each contactor can then be modeled using an integer (environment) variable $k_i$ for which we require: $\G_{[0, \infty)} (0 \leq k_i \leq 4)$.

\textbf{Guarantees.} \emph{If a generator becomes unavailable (fails), the controller shall disconnect it from the power network in 20 ms or less}. Let $g_0$ and $g_1$ be Boolean environment variables representing the state of the generators, where $1$ stands for ``available'' and $0$ for ``failure." We encode the above guarantees as $\G_{[0, \infty)}(g_i \rightarrow \F_{[0,1]}\tilde{c}_i).$ \emph{A DC bus shall never be disconnected from an AC generator for 100 ms or more,} i.e., $\G_{[0, \infty)}(\neg b_i \rightarrow \F_{[0,5]}b_i),$ where  $b_i$,  $i\in \{0,\ldots,3\}$ is a set of Boolean variables denoting the status of a bus, where $1$ stands for ``powered'' and $0$ for ``unpowered." Additional guarantees, which can also be expressed as STL formulas, include: (i) If both AC generators are available, the left AC generator shall power the left AC bus, and the right AC generator shall power the right AC bus. $C_3$ and $C_4$ shall be closed. (ii)  If one generator becomes unavailable, all buses shall be connected to the other generator. (iii)  Two generators must never be directly connected.
%\begin{itemize}
%\item If both AC generators are available, the top AC generator shall power the top AC bus, and the bottom AC generator shall power the bottom AC bus. $C_3$ and $C_4$ shall be closed.
%\item If one generator becomes unavailable, all buses shall be connected to the other generator.
%\item Two generators must never be connected.
%\end{itemize}

We apply the diagnosis and repair procedure in Section~\ref{sec:adv} to investigate whether there exists a control strategy that can satisfy the specification above over all possible values of contactor delays. As shown in Fig.~\ref{fig:EPS}, the controller is unrealizable; 
% the infeasibilty occurs at time $t=280$ ms, denoted by the thick black line. 
%  We see a game where the controller and environment are playing. In this case, 
a trace of contactor delays equal to $4$ at all times provides a counterexample, which leaves $DB_1$ unpowered for $160$~ms, exceeding the maximum allowed delay of $100$-ms. In fact, the controller cannot close $C_2$ until $C_1$ is tested as being open, to ensure that $G_1$ is safely isolated from $G_2$. 
% Fig.~\ref{fig:EPS} shows the counterexample returned by the first run of CEGIS loop where $k = 80$~ms. 
To guarantee realizability, Algorithm~\ref{algo:DiagnosisRepairAdversarial} suggests to either modify our assumptions to $\G_{[0, \infty)} (0 \leq k_i \leq 2)$ for $i\in\{0,\ldots,4\}$ or relax the guarantee on $DB_1$ to $G_{[0, \infty)}(\neg b_3 \rightarrow \F_{[0,8]} b_3)$.  
%\pierluigi{double-check second suggestion} 
The overall execution time was 326~s, which includes formulating and executing three CEGIS loops, requiring a total of 6 optimization problems. 

%%% Local Variables: 
%%% mode: latex
%%% TeX-master: "root"
%%% End: 

%Checking satisfiability : 
%Total time taken(problem formulation + yalmip) = 79.5950s 
%Gurobi time  = 0.9371
%
%1st CEGIS run :
%Total time taken(problem formulation + yalmip) = 107.0224s 
%Gurobi time  = 0.1298
%
%2nd CEGIS run :
%Total time taken(problem formulation + yalmip) = 95.1188s 
%Gurobi time  = 0.1107
%
%3rd CEGIS run :
%Total time taken(problem formulation + yalmip) = 44.1549s 
%Gurobi time  = 0.0598
%
%Thus total time = 325.8911 sec out of which 1.2374 sec is taken by the optimizer to solve roughly 6 optimization problems

\section{Conclusion}
\label{sec:conclusions}
We presented a set of algorithms for diagnosis and repair of STL specifications in the setting of controller synthesis for hybrid systems using a model predictive control scheme. 
% where the synthesis problem translates to solving a MILP. We alleviate the infeasibility of this MILP translation to indicate the reasons of unrealizability for the synthesis problem.
Given an unrealizable specification, our algorithms can detect possible reasons for infeasibility and suggest repairs to make it realizable.
% We prove our algorithm is sound and complete, and 
We showed the effectiveness of our approach on the synthesis of controllers for several applications. 
% including autonomous driving, a quadrotor system, and an aircraft electric power system. 
As future work, we plan to investigate techniques that better leverage the structure of the STL formulas and extend to a broader range of environment assumptions in the adversarial setting.  
% The current framework can fix an STL by fixing either the timing or the predicates. We haven't yet done a full study of how to combine the two. 

\section{Acknowledgments}
This work was partially supported by IBM and United Technologies Corporation (UTC) via the iCyPhy consortium, and by TerraSwarm, one of six centers of STARnet, a Semiconductor Research Corporation program sponsored by MARCO and DARPA.

%{\small
\bibliographystyle{abbrv}
\bibliography{hsccrefs}
%}

\end{document}